\documentclass{amsart}
\long\def\sidebyside#1#2{%
 \hbox to\textwidth{\vtop{\hsize=.5\textwidth%
 \advance\hsize by -.5\columnsep
\parindent=0pt
\centering

 #1\vskip1sp}\hskip\columnsep\vtop{\hsize=.5\textwidth%
 \advance\hsize by -.5\columnsep
\parindent=0pt
\centering
#2

}\hfill}}

\usepackage{amsthm, amssymb, amsfonts}
\usepackage{psfrag}
\usepackage{graphicx}
\newtheorem{Thm}{Theorem}

\newtheorem{ddefinition}[Thm]{Definition}
\newtheorem{ttheorem}[Thm]{Theorem}

 \bibliography{/home/ungar/dir_amy/dir_papers/dir_mybook01/dir_mythomas/ungar,/home/ungar/dir_amy/dir_papers/dir_mybook01/dir_mythomas/refs01,/home/ungar/dir_amy/dir_papers/dir_mybook01/dir_mythomas/refgeneral}
\newcommand {\vi}{\mathbb{V}}
\newcommand {\ccdot}{\mathbf{\cdot }} % cd has spacing nicer than cdot
\newcommand {\ub}{\mathbf{u}}
\newcommand {\vb}{\mathbf{v}}
\newcommand {\wb}{\mathbf{w}}
\newcommand {\vc}{\mathbb{V}_{\!c}}

\newcommand {\unpuvc}{ \lower.6ex \hbox {$1 + \uvc$} }
\newcommand {\op}{\mathbf{\oplus}} % op has spacing nicer than oplus
\newcommand {\om}{\mathbf{\ominus}} % om has spacing nicer than oplus
\newcommand {\od}{\mathbf{\otimes}}   % od has spacing nicer than oplus
\newcommand {\uvc}{\displaystyle\frac{\lower.6ex \hbox {$\ub\ccdot\vb$}}{c^2}}
\newcommand {\gub}{\gamma_{\ub}^{\phantom{1}}}
\newcommand {\gvb}{\gamma_{\vb}^{\phantom{1}}}
\newcommand {\gupvb}{\gamma_{\ub\op\vb}^{\phantom{1}}}
\newcommand {\gubs}{\gamma_{\ub}^2}
\newcommand {\gvbs}{\gamma_{\vb}^2}

\newcommand {\Rtwo}{\Rb^2}
\newcommand {\Rctwou}{{\Rb}_{c=1}^2}
\newcommand {\Rstwou}{{\Rb}_{s=1}^2}
\newcommand {\Rt}{\Rb^3}
\newcommand {\Rb}{\mathbb{R}}
\newcommand {\Rct}{{\Rb}_{c}^{3}}
\newcommand {\Rctwo}{{\Rb}_{c}^{2}}

\newcommand {\Rcn}{{\Rb}_{c}^{n}}
\newcommand {\Rstwo}{{\Rb}_{s}^2}
\newcommand {\Rsn}{{\Rb}_{s}^{n}}
\newcommand {\Rn}{\Rb^n}
\newcommand {\bz}{\mathbf{0}}
\newcommand {\zerb}{\mathbf{0}}
\newcommand {\inn}{\hspace{-0.1cm}\in\hspace{-0.1cm}}
\newcommand {\gyr}{{\rm gyr}}
\newcommand {\gyruvb}{\gyr[\ub,\vb]}

\newcommand {\gyrvub}{\gyr[\vb,\ub]}
\newcommand {\VS}{\mathbb{V}_s}
\newcommand {\ab}{\mathbf{a}}
\newcommand {\bb}{\mathbf{b}}
\newcommand {\cb}{\mathbf{c}}
\newcommand {\db}{\mathbf{d}}
\newcommand {\Aut}{{\rm Aut}}

\newcommand {\gyrab}{\gyr[a,b]}

\newcommand {\sqp}{\boxplus}
\newcommand {\sqm}{\boxminus}
\newcommand {\ro}{r_{_1}}
\newcommand {\rt}{r_{_2}}
\newcommand {\pb}{\mathbf{p}}

\newcommand {\Pb}{\mathbf{P}}
\newcommand {\xb}{\mathbf{x}}

\newcommand {\half}{\textstyle\frac{1}{2}}
\newcommand {\LAB}{L_{^{AB}}^{\phantom{o}}}
\newcommand {\MAB}{M_{^{AB}}^{\phantom{o}}}
 \newcommand {\MAD}{M_{^{AD}}^{\phantom{o}}}
 \newcommand {\MBC}{M_{^{BC}}^{\phantom{o}}}
 
\newcommand {\MABDC}{M_{^{ABDC}}^{\phantom{o}}}
 
 \newcommand {\gmA}{\gamma_{_A}}
 \newcommand {\gmB}{\gamma_{_B}}
 \newcommand {\gmC}{\gamma_{_C}}
 \newcommand {\gmD}{\gamma_{_D}}
\begin{document}
%%%%%%%%%%%%%%%%%%%%%%%%%%%%%%%%%%%%%%%%%%%%%%%%%%%%%%%%%%%%%%%%%%%%%%%%%%%%%
%
% TITLE
% TITLE  Einstein's Special Relativity: The Hyperbolic Geometric Viewpoint
% TITLE
% Paper for my talk in Budapest, 4-6 September 2009, as a Main Speaker.
%%%%%%%%%%%%%%%%%%%%%%%%%%%%%%%%%%%%%%%%%%%%%%%%%%%%%%%%%%%%%%%%%%%%%%%%%%%%%
%/space/home/ungar/dir_amy/dir_papers/dir_mybook01/dir_mythomas/dir_papers_tex/dir_paper057*
%%%%%%%%%%%%%%%%%%%%%%%%%%%%%%%%%%%%%%%%%%%%%%%%%%%%%%%%%%%%%%%%%%%%%%%%%%%%%
%\tableofcontents
%\newpage
%Deadline for Abstracts - May 20, 2009.
%\vskip.8cm
%$paper069\_Budapest2009.tex$
%{\rm ABC}
%{\curly{Abraham Albert Ungar}}
%\addtocounter{page}{-8}
%%%%%%%%%%%%%%%%%%%%%%%%%%%%%%%%%%%%%%%%%%%%%%%%%%%%%%%%%%%%%%%%%%%%%%%%%%%%%
% \pagebreak
%%%%%%%%%%%%%%%%%%%%%%%%%%%%%%%%%%%%%%%%%%%%%%%%%%%%%%%%%%%%%%%%%%%%%%%%%%%%%
%\pagenumbering{arabic}
%\addtocounter{page}{-1}
\begin{center}
% TITLE
\Large{
Einstein's Special Relativity: \\ The Hyperbolic Geometric Viewpoint
     }
\end{center}
%\vspace{-0.8cm}
\begin{center}
Abraham A. Ungar\\
Department of Mathematics\\
North Dakota State University\\
Fargo, ND 58105, USA\\
Email: abraham.ungar@ndsu.edu\\
%paper061_mybook03review.tex  \\
%%%%%%%%%%%%%%%%%%%%%%%%%%%%%%%%%%%%%%%%%%%%%%%%%%%%%%%%%%%%%%%%%%%%%%%%%%%%%
May 2, 2009 \\[12pt]
Conference on Mathematics, Physics and Philosophy \\
on the Interpretations of Relativity, II, \\
Budapest, 4-6 September, 2009 \\[12pt]
Published in:\\
PIRT Conference Proc., 4-6 Sept. 2009, Budapest, pages 1--35.
%%%%%%%%%%%%%%%%%%%%%%%%%%%%%%%%%%%%%%%%%%%%%%%%%%%%%%%%%%%%%%%%%%%%%%%%%%%%%
\end{center}

%%%%%%%%%%%%%%%%%%%%%%%%%%%%%%%%%%%%%%%%%%%%%%%%%%%%%%%%%%%%%%%%%%%%%%%%%%%%%
\begin{quotation}
{\bf ABSTRACT}\phantom{OO}
The analytic hyperbolic geometric viewpoint of Einstein's special theory of relativity
is presented. Owing to the introduction of vectors into hyperbolic geometry,
where they are called {\it gyrovectors}, the use of analytic hyperbolic geometry
extends Einstein's unfinished symphony significantly, elevating it to the status
of a mathematical theory that could be emulated to the benefit of the entire
mathematical and physical community.
The resulting theory involves a gyrovector space approach to hyperbolic geometry
and relativistic mechanics, and
could be studied with profit by anyone with a sufficient background
in the common vector space approach to Euclidean geometry and classical mechanics.
Einstein noted in his 1905 paper that founded the special theory of relativity
that his velocity addition law satisfies the law of velocity parallelogram
only to a first approximation.
Within our hyperbolic geometric viewpoint of special relativity it becomes clear
that Einstein's velocity addition law leads to a hyperbolic parallelogram addition law
of Einsteinian velocities, which is supported experimentally by the cosmological effect
known as stellar aberration and its relativistic interpretation.
The latter, in turn, is supported experimentally by the
``GP-B'' gyroscope experiment developed by NASA and Stanford University.
Furthermore, the hyperbolic viewpoint of special relativity
meshes extraordinarily well with the Minkowskian four-vector formalism of special relativity,
revealing that the seemingly notorious relativistic mass
meshes up with the four-vector formalism as well, owing to the
natural emergence of {\it dark matter}.
It is therefore hoped that both special relativity and its underlying
analytic hyperbolic geometry will become part of the  lore learned by all
undergraduate and graduate mathematics and physics students.
\end{quotation}
%%%%%%%%%%%%%%%%%%%%%%%%%%%%%%%%%%%%%%%%%%%%%%%%%%%%%%%%%%%%%%%%%%%%%%%%%%%%%


%SECTION 1
\section{Introduction}
\label{secc1}

It is a pleasure for me to be given this opportunity to participate in this
Conference on Mathematics, Physics and Philosophy in the Interpretations
of relativity, II, in order to present
the novel hyperbolic geometric viewpoint of Einstein's Special Theory of Relativity.

The hyperbolic geometry of Bolyai and Lobachevsky underlies relativistic physics
just as Euclidean geometry underlies classical physics.
As such, it suggests
a rather radical break of the traditional study of
Einstein's special theory of relativity, required for the restoration of the
glory and harmony of special relativity in the Twenty-First Century.
The mathematical structure that Einstein's velocity addition law encodes
emerges in all its splendor, leading to a
nocommutative-nonassociative algebraic setting for the
Beltrami-Klein ball model of hyperbolic geometry.

Einstein's velocity addition law is non-commutative (and non-associative, as well).
Indeed, Einstein's expos\'e of velocity composition for two inertial
systems emphasizes the lack of symmetry in the formula for the direction
of the relative composite velocity vector
\cite[pp.~905\,--\,906]{einstein05}
\cite[p.~117]{walter99b}.
\'Emile Borel's attempt to ``repair'' the seemingly ``defective'' Einstein velocity
addition in the years following 1912 is described by Walter in
\cite[p.~117]{walter99b}:
``Borel could construct a tetrahedron in kinematic space, and determined thereby both
the direction and magnitude of relative [composite] velocity in a symmetric manner.''

The goal of this lecture is to demonstrate that Einstein's velocity addition law
possesses rich structure and, hence, should be placed centrally in special relativity theory
along with the algebraic structure that it encodes.
We will find that Einstein's velocity addition law encodes a
vector space-like structure, called an {\it Einstein gyrovector space}, which
forms the setting for the Beltrami-Klein ball model of hyperbolic geometry
just as vector spaces form the setting for the standard model of Euclidean geometry.

Accordingly, placing Einstein's velocity addition law centrally in the basis of
Einstein's special theory of relativity amounts to the study of the theory by means
of its underlying analytic hyperbolic geometry.
The resulting hyperbolic geometric viewpoint of Einstein's special theory of relativity
is rewarding. It suggests, for instance, the introduction of vectors into
hyperbolic geometry, where they are called {\it gyrovectors}, and the study
of Einsteinian velocities as gyrovectors that add according to the
gyroparallelogram addition law.

Indeed, in the years 1908\,--\,1914, the period which experienced a dramatic flowering of
creativity in the special theory of relativity, the Croatian physicist and mathematician
Vladimir Vari\v cak (1865\,--\,1942), professor and rector of Zagreb
University, showed that this theory has a natural interpretation in
hyperbolic geometry \cite{varicak10,barrett01}.
However, much to his chagrin, he had to admit in 1924 \cite[p.~80]{varicak24}
that the adaption of
vector algebra for use in hyperbolic geometry was just not feasible,
as Scott Walter notes in \cite[p.~121]{walter99b}.
Vladimir Vari\v cak's hyperbolic geometry program, cited by
Pauli \cite[p.~74]{pauli}, is described by Walter in
\cite[p.~112--115]{walter99b}.

We will find in this lecture that Borel's attempt to ``repair''
Einstein's velocity addition law, and Vari\v cak's concern about the lack of
vector algebra in hyperbolic geometry are both unjustified. Indeed,
we will find that Einstein's velocity addition law admits a gyrovector space structure,
in which Einsteinian velocities are gyrovectors that add commutatively
according to the gyroparallelogram addition law.
The prefix ``gyro'' that we use extensively to capture analogies stems from
{\it Thomas gyration}, which is the mathematical abstraction of the
special relativistic effect known as {\it Thomas precession}.
In fact, it is the mere introduction of Thomas gyration that turns
Euclidean geometry into hyperbolic geometry, and
some results of classical mechanics into corresponding results in relativistic mechanics.

%SECTION 2
\section{Einstein Velocity Addition}
\label{secc2}

Let $c$ be any positive constant, let $(\vi,+,\ccdot)$ be any real inner product space,
and let
\begin{equation} \label{eqcball}
\vc  = \{\vb\in\vi: \|\vb\| < c \}
\end{equation}
be the $c$-ball of all relativistically admissible velocities of material
particles. It is the open ball of radius $c$, centered at the
origin of the real inner product space $\vi$, consisting of all vectors $\vb$
in $\vi$ with magnitude $\|\vb\|$ smaller than $c$.

Einstein velocity addition in the $c$-ball of all relativistically admissible velocities
is given by the equation
\cite[Eq.~2.9.2]{urbantkebookeng},\cite[p.~55]{moller52},\cite{fock},
\cite{mybook01},
\begin{equation} \label{eq01}
{\ub}\op{\vb}=\frac{1}{\unpuvc}
\left\{ {\ub}+ \frac{1}{\gub}\vb+\frac{1}{c^{2}}\frac{\gamma _{{\ub}}}{%
1+\gamma _{{\ub}}}( {\ub}\ccdot{\vb}) {\ub} \right\}
\end{equation}
satisfying the {\it gamma identity}
\begin{equation} \label{eqgupv00}
\gupvb = \gub\gvb\left(1+\frac{\ub\ccdot\vb}{c^2}\right)
\end{equation}
for all $\ub,\vb\in\vc $,
where $\gub$ is the gamma factor
\begin{equation} \label{v72gs}
\gvb = \frac{1}{\sqrt{1-\displaystyle\frac{\|\vb\|^2}{c^2}}}
\end{equation}
in the $c$-ball $\vc $.
Here $\ub\ccdot\vb$ and $\|\vb\|$ represent the inner product and the norm
that the ball $\vc$ inherits from its space $\vi$.

A frequently used identity that follows from \eqref{v72gs} is
\begin{equation} \label{tksnw}
\frac{\vb^2}{c^2} = \frac{\gvbs-1}{\gvbs}
\end{equation}
where we use the notation $\vb^2=\vb\ccdot\vb=\|\vb\|^2$.

In physical applications, $\vi=\Rt$ is the Euclidean 3-space,
which is the space of all classical, Newtonian velocities, and
$\vc=\Rct\subset\Rt$ is the $c$-ball of $\Rt$ of all relativistically
admissible, Einsteinian velocities.
Furthermore, the constant $c$ represents in physical applications the
vacuum speed of light.

Einstein addition \eqref{eq01} of relativistically
admissible velocities was introduced by Einstein in his 1905 paper
\cite{einstein05} \cite[p.~141]{einsteinfive}
that founded the special theory of relativity,
where the magnitudes
of the two sides of Einstein addition \eqref{eq01} are presented.
One has to remember here that the Euclidean 3-vector algebra was not so
widely known in 1905 and, consequently, was not used by Einstein.
Einstein calculated in \cite{einstein05} the behavior
of the velocity components parallel and orthogonal to the relative
velocity between inertial systems, which is as close as one can get
without vectors to the vectorial version \eqref{eq01}.

We naturally use the abbreviation
$\ub\om\vb=\ub\op(-\vb)$ for Einstein subtraction, so that,
for instance, $\vb\om\vb = \zerb$,
$\om\vb = \zerb\om\vb=-\vb$ and, in particular,
\begin{equation} \label{eq01a}
\om(\ub\op\vb) = \om\ub\om\vb
\end{equation}
and
\begin{equation} \label{eq01b}
\om\ub\op(\ub\op\vb) = \vb
\end{equation}
for all $\ub,\vb$ in the ball,
in full analogy with vector addition and subtraction. Identity
\eqref{eq01a} is known as the {\it automorphic inverse property},
and Identity
\eqref{eq01b} is known as the {\it left cancellation law} of Einstein
addition \cite{mybook02,mybook03,mybook04}.
We may note that
Einstein addition does not obey the immediate right counterpart of the
left cancellation law \eqref{eq01b} since, in general,
\begin{equation} \label{eq01c}
(\ub\op\vb)\om\vb \ne \ub
\end{equation}
However, this seemingly lack of a right cancellation law will be repaired
in \eqref{eq01bb}, following the suggestive introduction of a second
gyrogroup binary operation in Def.~\ref{defdual} below, which
captures analogies.

In the Newtonian limit of large $c$, $c\rightarrow\infty$, the ball $\vc $
expands to the whole of its space $\vi$, as we see from \eqref{eqcball},
and Einstein addition $\op$ in $\vc $
reduces to the ordinary vector addition $+$ in $\vi$,
as we see from \eqref{eq01} and \eqref{v72gs}.

Einstein addition is noncommutative \cite{barrett07}.
While $\|\ub\op\vb\|=\|\vb\op\ub\|$, we have, in general,
\begin{equation} \label{eqyt01}
\ub\op\vb\ne\vb\op\ub
\end{equation}
$\ub,\vb\in\vc $. Moreover, Einstein addition is also nonassociative
since, in general,
\begin{equation} \label{eqyt02}
(\ub\op\vb)\op\wb\ne\ub\op(\vb\op\wb)
\end{equation}
$\ub,\vb,\wb\in\vc $.

It seems that following the breakdown of
commutativity and associativity in Einstein addition some mathematical
regularity has been lost in the transition from
Newton's velocity vector addition in $\vi$ to
Einstein's velocity addition \eqref{eq01} in $\vc $. This is, however, not the
case since, as we will see in Sec.~\ref{secc3},
Thomas gyration comes to the rescue
\cite{mybook01,mybook02,mybook03,mybook04,walterrev2002,rassiasrev2008}.
Indeed, we will find in Sec.~\ref{secc3} that the mere introduction of gyrations endows
the Einstein groupoid $(\vc,\op)$ with a grouplike rich structure \cite{grouplike}
that we call a {\it gyrogroup}.

When the nonzero vectors $\ub,\vb\in\vc\subset\vi$ are parallel
in $\vi$, $\ub \| \vb$, that is, $\ub=\lambda\vb$
for some $0\ne\lambda\in\Rb$, Einstein addition reduces to the Einstein
addition of parallel velocities \cite[p.~50]{whittaker49},
\begin{equation} \label{eq1pfck03}
\ub\op\vb = \frac{\ub+\vb}{1+\frac{1}{c^2}\|\ub\|\|\vb\|}, \qquad
\ub \| \vb
\end{equation}
which was confirmed experimentally by the Fizeau's 1851 experiment \cite{miller81}.
Owing to its simplicity, some books on special relativity present
Einstein velocity addition in its restricted form \eqref{eq1pfck03}
rather than its general form \eqref{eq01}.

The restricted Einstein addition \eqref{eq1pfck03}
is both commutative and associative.
Accordingly, the restricted Einstein addition is a group operation,
as Einstein noted in \cite{einstein05}; see \cite[p.~142]{einsteinfive}.
In contrast, Einstein made no remark about group properties of his addition
of velocities that need not be parallel. Indeed, the general Einstein
addition \eqref{eq01} is not a group operation but, rather, a gyrocommutative
gyrogroup operation, a structure that was discovered more than
80 years later, in 1988 \cite{parametrization}, and is presented in
Def.~\ref{defroupx} in Sec.~\ref{secc3}.

%SECTION 3
\section{Thomas Gyration and Einstein Gyrogroups}
\label{secc3}

For any $\ub,\vb\inn\vc $, let $\gyruvb : \vc \rightarrow\vc $ be the self-map of $\vc $
given in terms of Einstein addition $\op$ by the equation \cite{parametrization}
\begin{equation} \label{eq004}
\gyruvb \wb = \om(\ub\op\vb)\op\{\ub\op(\vb\op\wb)\}
\end{equation}
where $\om\vb=-\vb$, for all $\wb\inn\vc $.
The self-map $\gyruvb$ of $\vc $, which takes $\wb\inn\vc $ into
$\om(\ub\op\vb)\op\{\ub\op(\vb\op\wb)\inn\vc $,
is called the {\it Thomas gyration} generated by $\ub$ and $\vb$.
Thomas gyration is the mathematical abstraction of the relativistic effect
known as Thomas precession \cite[Chap.~1]{mybook01},\cite[Sec.~10.3]{mybook03},
and it has an interpretation in hyperbolic geometry
\cite{vermeer05}
as the negative hyperbolic triangle defect \cite[Theorem 8.55]{mybook03}.

In the Newtonian limit, $c\rightarrow\infty$, Einstein addition $\op$ in $\VS$
reduces to the common vector addition + in $\vi$, which is associative.
Accordingly, in this limit the gyration $\gyruvb$ in \eqref{eq004} reduces to the
identity map of $\vi$. Hence, as expected, Thomas gyrations
$\gyruvb$, $\ub,\vb\inn\vc $,
vanish (that is, they become {\it trivial}) in the Newtonian limit.

It is clear from the gyration equation \eqref{eq004} that gyrations
measure the extent to which Einstein addition is nonassociative, where associativity
corresponds to trivial gyrations.

The gyration equation \eqref{eq004} can be manipulated
(with the help of computer algebra) into the equation
\begin{equation} \label{hdge1ein}
\gyruvb\wb = \wb + \frac{A\ub+B\vb}{D}
\end{equation}
where
%%%%%%%%%%%%%%%%%%%%%%%%%%%%%%%%%%%%%%%%%%%%%%%%%%%%%%%%%%%%%%%%%%%%
\begin{equation} \label{hdgej2ein}
\begin{split}
 A &=-\frac{1}{c^2}\frac{\gubs}{(\gub+1)} (\gvb-1) (\ub\ccdot\wb)
 +
 \frac{1}{c^2}\gub\gvb (\vb\ccdot\wb)
\\[8pt] & \phantom{=} ~+
 \frac{2}{c^4} \frac{\gubs\gvbs}{(\gub+1)(\gvb+1)} (\ub\ccdot\vb) (\vb\ccdot\wb)
\\[8pt]
B &=- \frac{1}{c^2}
\frac{\gvb}{\gvb+1}
\{\gub(\gvb+1)(\ub\ccdot\wb) + (\gub-1)\gvb(\vb\ccdot\wb) \}
\\[8pt]
D &= \gub\gvb(1+ \frac{\ub\ccdot\vb}{c^2}) +1 = \gamma_{\ub\op\vb}^{\phantom{O}} + 1 > 1
\end{split}
\end{equation}
% MATHEMATICA "einstein03",  MATLAB "gyren.m", "test0423".
%%%%%%%%%%%%%%%%%%%%%%%%%%%%%%%%%%%%%%%%%%%%%%%%%%%%%%%%%%%%%%%%%%%%
for all $\ub,\vb,\wb\in\vc $.
Allowing $\wb\in\vi\supset\vc $ in \eqref{hdge1ein}\,--\,\eqref{hdgej2ein},
gyrations $\gyr[\ub,\vb]$
are expendable to linear maps of $\vi$ for all $\ub,\vb\in\vc $.

In each of the three special cases when
(i) $\ub=\zerb$, or
(ii) $\vb=\zerb$, or
(iii) $\ub$ and $\vb$ are parallel
in $\vi$, $\ub\|\vb$, we have
$A\ub+B\vb=\zerb$ so that $\gyr[\ub,\vb]$ is trivial,
\begin{equation} \label{sprdhein}
\begin{split}
\gyr[\zerb,\vb]\wb &= \wb  \\
\gyr[\ub,\zerb]\wb &= \wb \\
\gyr[\ub,\vb]\wb &= \wb , \hspace{1.2cm} \ub\|\vb
\end{split}
\end{equation}
for all $\ub,\vb\in\vc$, and all $\wb\in\vi$.

It follows from \eqref{hdge1ein} that
\begin{equation} \label{eq1ffmznein}
\gyr[\vb,\ub](\gyr[\ub,\vb]\wb) = \wb
\end{equation}
for all $\ub,\vb\in\vc $, $\wb\in\vi$, so that gyrations are invertible
linear maps of $\vi$,
the inverse of $\gyr[\ub,\vb]$ being $\gyr[\vb,\ub]$
for all $\ub,\vb\in\vc $.

Gyrations keep the inner product of
elements of the ball $\vc $ invariant, that is,
\begin{equation} \label{eq005}
%  \|\gyruvb \wb\| = \|\wb\|
\gyruvb\ab\ccdot\gyruvb\bb = \ab\ccdot\bb
\end{equation}
for all $\ab,\bb,\ub,\vb\inn\vc $. Hence, $\gyruvb$ is an
{\it isometry} of $\vc $,
keeping the norm of elements of the ball $\vc $ invariant,
\begin{equation} \label{eq005a}
\|\gyruvb \wb\| = \|\wb\|
\end{equation}
Accordingly, $\gyruvb$ represents a rotation of the ball $\vc $ about its origin for
any $\ub,\vb\inn\vc $.

The invertible self-map $\gyruvb$ of $\vc $ respects Einstein addition in $\vc $,
\begin{equation} \label{eq005b}
\gyruvb (\ab \op \bb) = \gyruvb\ab \op \gyruvb\bb
\end{equation}
for all $\ab,\bb,\ub,\vb\inn\vc $,
so that $\gyruvb$ is an automorphism of the Einstein groupoid $(\vc ,\op)$.
We recall that an automorphism of a groupoid $(\vc ,\op)$
is a bijective self-map of the groupoid $\vc $
that respects its binary operation, that is, it satisfies \eqref{eq005b}.
Under bijection composition the automorphisms
of a groupoid $(\vc ,\op)$
form a group known as the automorphism group, and denoted
$\Aut(\vc ,\op)$.
Being special automorphisms, Thomas gyrations
$\gyruvb \inn \Aut(\vc ,\op)$, $\ub,\vb\inn\vc $,
are also called {\it gyroautomorphisms}, $\gyr$ being the gyroautomorphism
generator called the {\it gyrator}.

The gyroautomorphisms $\gyruvb$ regulate Einstein addition in the ball $\vc $,
giving rise to the following nonassociative algebraic laws that ``repair''
the breakdown of commutativity and associativity in Einstein addition:
%%%%%%%%%%%%%%%%%%%%%%%%%%%%%%%%%%%%%%%%%%%%%%%%%%%%%%%%%%%%%%%%%%%%%%%%%%%%%
 \begin{alignat}{2}\label{laws00}
 \notag
  \ub\op\vb & \!=\! \gyruvb(\vb\op\ub) &&\hspace{0.8cm}\text{Gyrocommutativity}\\
 \notag
  \ub\op(\vb\op\wb)& \!=\! (\ub\op\vb)\op\gyruvb\wb&&\hspace{0.8cm}
  \text{Left Gyroassociativity} \\
 \notag
  (\ub\op\vb)\op\wb& \!=\! \ub\op(\vb \op\gyrvub\wb) &&\hspace{0.8cm}
  \text{Right Gyroassociativity} \\
 \end{alignat}
%%%%%%%%%%%%%%%%%%%%%%%%%%%%%%%%%%%%%%%%%%%%%%%%%%%%%%%%%%%%%%%%%%%%%%%%%%%%%
for all $\ub,\vb,\wb\inn\vc $.

Owing to the gyrocommutative law in \eqref{laws00},
Thomas gyration is recognized as the familiar Thomas precession.
The gyrocommutative law was already known to Silberstein in 1914
\cite{silberstein14} in the following sense.
The Thomas precession generated by $\ub,\vb\inn\Rct$
% Indeed, for the sake of uniqueness below I need the "Rct" (not "VS") above
is the unique rotation
that takes $\vb\op\ub$ into $\ub\op\vb$ about an axis perpendicular to the
plane of $\ub$ and $\vb$ through an angle $< \pi$ in $\vi$,
thus giving rise to the gyrocommutative law.
Obviously, Silberstein did not use the terms
``Thomas precession'' and ``gyrocommutative law'' since
these terms have been coined
later, respectively, following Thomas' 1926 paper \cite{thomas26},
and in 1991 \cite{grouplike,gaxioms}.
%\textcolor{red}{
We may remark that
Thomas precession has purely kinematical origin,
as emphasized in \cite{ungarthomas06}. Accordingly, the presence of
Thomas precession is not connected with the action of any force.
%}

Contrasting the discovery before 1914 of what we presently call
the gyrocommutative law, the gyroassociative laws, left and right,
were discovered about 75 years later, in 1988 \cite{parametrization}.

A most important property of Thomas gyration is the so called
{\it loop property} (left and right),
\begin{equation} \label{loops00}
\begin{array}{cl}
{\gyr}[\ub\op\vb,\vb]=\gyruvb
& \hbox{ \ \ \ \ \ \ Left Loop Property} \\[3pt]
{\gyr}[\ub,\vb\op\ub]=\gyruvb
& \hbox{ \ \ \ \ \ \ Right Loop Property} \\
\end{array}
\end{equation}
for all $\ub,\vb\inn\vc $.
The left loop property will prove useful in \eqref{eqtghj04}
in solving a basic gyrogroup equation.

The grouplike groupoid $(\vc ,\op)$ that regulates Einstein addition,
$\op$, in the ball $\vc $ of the Euclidean 3-space $\vi$ is a
{\it gyrocommutative gyrogroup} called an {\it Einstein gyrogroup}.
Einstein gyrogroups and gyrovector spaces are studied in
\cite{mybook01,mybook02,mybook03,mybook04}.
Gyrogroups are not peculiar to Einstein addition \cite{mbtogyp08}.
Rather, they are abound in the theory of groups \cite{tuvalungar01,tuvalungar02,feder03},
loops \cite{issa99},
quasigroup \cite{issa2001,kuznetsov03},
and Lie groups \cite{kasparian04,kikkawa75,kikkawa99}.

Taking the key features of Einstein velocity addition law, and guided by analogies
with groups, we are led to the following formal definition of
abstract gyrogroups:

% DEFINITION NUMBER 1
\begin{ddefinition}\label{defroupx}
{\bf (Gyrogroups).}
{\it
A groupoid is a non-empty set with a binary operation.
A groupoid $(G , \op )$
is a gyrogroup if its binary operation satisfies the following axioms.
In $G$ there is at least one element, $0$, called a left identity, satisfying

\noindent
(G1) \hspace{1.2cm} $0 \op a=a$

\noindent
for all $a \inn G$. There is an element $0 \inn G$ satisfying axiom $(G1)$ such
that for each $a\inn G$ there is an element $\om a\inn G$, called a
left inverse of $a$, satisfying

\noindent
(G2) \hspace{1.2cm} $\om a \op a=0\,.$

\noindent
Moreover, for any $a,b,c\inn G$ there exists a unique element $\gyr[a,b]c \inn G$
such that the binary operation obeys the left gyroassociative law

\noindent
(G3) \hspace{1.2cm} $a\op(b\op c)=(a\op b)\op\gyrab c\,.$

\noindent
The map $\gyr[a,b]:G\to G$ given by $c\mapsto \gyr[a,b]c$
is an automorphism of the groupoid $(G,\op)$, that is,

\noindent
(G4) \hspace{1.2cm} $\gyrab\inn\Aut (G,\op) \,,$

\noindent
and the automorphism $\gyr[a,b]$ of $G$ is called
the gyroautomorphism, or the gyration, of $G$ generated by $a,b \inn G$.
The operator $\gyr : G\times G\rightarrow\Aut (G,\op)$ is called the
gyrator of $G$.
Finally, the gyroautomorphism $\gyr[a,b]$ generated by any $a,b \inn G$
possesses the left loop property

\noindent
(G5) \hspace{1.2cm} $\gyrab=\gyr [a\op b,b] \,.$
}
\end{ddefinition}

The first pair of the gyrogroup axioms are like the group axioms.
The last pair present the gyrator axioms and the middle axiom links the two pairs.

As in group theory, we use the notation
$a \om b = a \op (\om b)$
in gyrogroup theory as well.

Some groups are commutative.  In full analogy, some gyrogroups are gyrocommutative.

% DEFINITION NUMBER 3
\begin{ddefinition}\label{defgyrocomm}
{\bf (Gyrocommutative Gyrogroups).}
{\it
A gyrogroup $(G, \oplus )$ is gyrocommutative if
its binary operation obeys the gyrocommutative law

\noindent
(G6) \hspace{1.2cm} $a\oplus b=\gyrab(b\oplus a)$

\noindent
for all $a,b\inn G$.
}
\end{ddefinition}

We thus see that it is Thomas gyration that
regulates Einstein addition, endowing it with the
rich structure of a gyrocommutative gyrogroup.

The gyrogroup operation (or, addition) of any gyrogroup
has an associated dual operation, called the {\it gyrogroup cooperation}
(or, {\it coaddition}), the definition of which follows:

%%%%%%%%%%%%%%%%%%%%%%%%%%%%%%%%%%%%%%%%%%%%%%%%%%%%%%%%%%%%%%%%%%%%
% DEFINITION NUMBER 2.9
\begin{ddefinition}\label{defdual}
{\bf (The Gyrogroup Cooperation (Coaddition)).}
{\it
Let $(G,\op)$ be a gyrogroup with gyrogroup operation (or, addition) $\op$.
The gyrogroup cooperation (or, coaddition) $\sqp$ is a second
binary operation in $G$ given by the equation
\begin{equation} \label{eqdfhn01}
a\sqp b = a \op\gyr[a,\om b]b
\end{equation}
for all $a,b\in G$.
}
\end{ddefinition}

Replacing $b$ by $\om b$ in \eqref{eqdfhn01}
we have the {\it cosubtraction} identity
\begin{equation} \label{eqdfhn01b}
a\sqm b := a\sqp(\om b) = a \om\gyr[a,b]b
\end{equation}
for all $a,b\in G$.

To motivate the introduction of the gyrogroup cooperation, let us solve
the equation
\begin{equation} \label{eqtghj03}
x\op a = b
\end{equation}
for the unknown $x$ in a gyrogroup $(G,\op)$.

Assuming that a solution
$x$ exists, we have the following chain of equations
%%%%%%%%%%%%%%%%%%%%%%%%%%%%%%%%%%%%%%%%%%%%%%%%%%%%%%%%%%%%%%%%%%%%
 \begin{equation} \label{eqtghj04}
 \begin{split}
 x &= x \op 0 \\
&= x \op( a \om  a )\\
&=( x \op  a )\op\gyr [ x , a ](\om  a )\\
&=( x \op  a )\om\gyr [ x , a ] a \\
&=( x \op  a )\om\gyr [ x \op  a , a ] a \\
&=b\om\gyr[b,a]a \\
&=b\sqm a
 \end{split}
 \end{equation}
%%%%%%%%%%%%%%%%%%%%%%%%%%%%%%%%%%%%%%%%%%%%%%%%%%%%%%%%%%%%%%%%%%%%
where the gyrogroup cosubtraction, \eqref{eqdfhn01b},
which captures here an obvious analogy, comes into play.
Hence, if a solution $x$ to the gyrogroup equation \eqref{eqtghj03}
exists, it must be given uniquely by \eqref{eqtghj04}. One can finally show
that the latter is indeed a solution \cite[Sec.~2.4]{mybook03}.

The gyrogroup cooperation is introduced into gyrogroups in order
to  capture useful analogies between gyrogroups and groups, and
to uncover duality symmetries
with the gyrogroup operation.
Thus, for instance, the gyrogroup cooperation recovers the seemingly
missing right counterpart of the left cancellation law \eqref{eq01b},
giving rise to the following right cancellation law:
\begin{equation} \label{eq01bb}
(\vb\op\ub) \sqm \ub = \vb
\end{equation}
for all $\ub,\vb$ in the ball.
Furthermore, the right cancellation law \eqref{eq01bb} can be dualized,
giving rise to the dual right cancellation law
\begin{equation} \label{eq01bbb}
(\vb\sqp\ub) \om \ub = \vb
\end{equation}

As an example, and for later reference, we note that it follows from
\eqref{eq01bb} that
\begin{equation} \label{hufns}
\db = (\bb\sqp \cb)\om\ab  \hspace{0.6cm} \Longrightarrow \hspace{0.6cm}
\bb\sqp \cb = \db\sqp\ab 
\end{equation}
in any gyrocommutative gyrogroup.

A gyrogroup cooperation is commutative if and only if the gyrogroup is
gyrocommutative
\cite[Theorem 3.4]{mybook02}
\cite[Theorem 3.4]{mybook03}.
Hence, in particular, Einstein coaddition is commutative.
Indeed, Einstein coaddition, $\sqp$, in an Einstein gyrogroup $(\VS,\op)$
is given by the equation
\cite[Eq.~3.195]{mybook03}
\begin{equation} \label{eincoadd}
\ub \sqp \vb = 2 \od \frac{\gub\ub+\gvb\vb}{\gub+\gvb}
\end{equation}
which is commutative, as expected.
Hence, for instance, $\db\sqp\ab=\ab\sqp\db$ in \eqref{hufns}.
The symbol $\od$ in \eqref{eincoadd} represents {\it scalar multiplication} so that,
for instance, $2\od\vb=\vb\op\vb$, for all $\vb$ in a gyrogroup $(G,\op)$,
as explained in Sec.~\ref{secc4}.

% SECTION NUMBER 4
\section{Einstein Gyrovector Spaces}\label{secc4}

Einstein addition in the ball admits scalar multiplication, giving rise to the
following definition \cite{gaxioms}.

%%%%%%%%%%%%%%%%%%%%%%%%%%%%%%%%%%%%%%%%%%%%%%%%%%%%%%%%%%%%%%%%%%%%
% DEFINITON NUMBER 5
\begin{ddefinition}\label{defgvspace}
An Einstein gyrovector space $(\VS,\op,\od)$
is an Einstein gyrogroup $(\VS,\op)$, $\VS\subset\vi$, with scalar multiplication
$\od$ given by the equation
 \begin{equation} \label{eqmlt03}
 \begin{split}
r\od\vb&= s \frac
 {\left(1+\displaystyle\frac{\|\vb\|}{s}\right)^r
- \left(1-\displaystyle\frac{\|\vb\|}{s}\right)^r}
 {\left(1+\displaystyle\frac{\|\vb\|}{s}\right)^r
+ \left(1-\displaystyle\frac{\|\vb\|}{s}\right)^r}
\frac{\vb}{\|\vb\|}\\[8pt]
&= s \tanh( r\,\tanh^{-1}\frac{\|\vb\|}{s})\frac{\vb}{\|\vb\|}
 \end{split}
 \end{equation}
where $r$ is any real number, $r\in\Rb$,
$\vb\in\VS$, $\vb\ne\bz$, and $r\od\bz=\bz$, and with
which we use the notation $\vb\od  r=r\od\vb$.
\end{ddefinition}
%%%%%%%%%%%%%%%%%%%%%%%%%%%%%%%%%%%%%%%%%%%%%%%%%%%%%%%%%%%%%%%%%%%%

Einstein gyrovector spaces are studied in \cite[Sec.~6.18]{mybook03}
and \cite{mybook04}.
Einstein scalar multiplication does not distribute over Einstein addition,
but it possesses other properties of vector spaces. For any
positive integer $n$, and for all real numbers $\ro,\rt\in\Rb$, and
$\vb\in\VS$, we have
%%%%%%%%%%%%%%%%%%%%%%%%%%%%%%%%%%%%%%%%%%%%%%%%%%%%%%%%%%%%%%%%%%%%%%%%%%%%%
\begin{alignat}{2}\label{scalarprp}
\notag
 n\od\vb&=\vb\op\dots\op\vb &&\qquad\text{$n$ terms}\\[3pt]
\notag
 (\ro+\rt)\od\vb&=\ro\od\vb\op\rt\od\vb
 &&\qquad\text{Scalar Distributive Law}\\[3pt]
\notag
(\ro\rt)\od\vb&=\ro\od(\rt\od\vb)
 &&\qquad\text{Scalar Associative Law} \\[3pt]
 \notag
r\od(\ro\od\ab\op\rt\od\ab)&=r\od(\ro\od\ab)\op r\od(\rt\od\ab)
 &&\qquad\text{Monodistributive Law}
 \notag
\end{alignat}
%%%%%%%%%%%%%%%%%%%%%%%%%%%%%%%%%%%%%%%%%%%%%%%%%%%%%%%%%%%%%%%%%%%%%%%%%%%%%
in any Einstein gyrovector space $(\VS,\op,\od)$.

Any Einstein gyrovector space $(\VS,\op,\od)$ inherits an inner product
and a norm from its vector space $\vi$. These turn out to be invariant under gyrations,
that is,
\begin{equation} \label{eqwiuj01}
\begin{split}
\gyr[\ab,\bb]\ub \ccdot \gyr[\ab,\bb]\vb &= \ub \ccdot \vb \\[3pt]
\| \gyr[\ab,\bb]\vb \| &= \| \vb \|
\end{split}
\end{equation}
for all $\ab,\bb,\ub,\vb\inn\VS$.

Unlike vector spaces, Einstein gyrovector spaces $(\VS,\op,\od)$ do not
possess the distributive law since, in general,
\begin{equation} \label{eqwdkj01}
r\od(\ub\op\vb) \ne r\od\ub\op r\od\vb
\end{equation}
for $r\inn\Rb$ and $\ub,\vb\inn\VS$.
One might suppose that there is a price to pay in mathematical regularity
when replacing ordinary vector addition with Einstein addition,
but  this is not the case as demonstrated in \cite{mybook01,mybook02,mybook03},
and as noted by S.~Walter in \cite{walterrev2002}.

In full analogy with the common Euclidean distance function,
Einstein addition gives rise to the
{\it gyrodistance} function
\begin{equation} \label{tkcflsen}
d_\op(\ab,\bb) = \|\om\ab\op\bb\|
\end{equation}
that obeys the gyrotriangle inequality
\cite[Theorem 6.9]{mybook03}
\begin{equation} \label{rif01}
d_\op(\ab,\bb) \le d_\op(\ab,\pb) \op d_\op(\pb,\bb)
\end{equation}
for any $\ab,\bb,\pb\in\VS$ in an Einstein gyrovector space $(\VS,\op,\od)$.
The gyrodistance function is invariant under the group of motions of any
Einstein gyrovector space, that is, under {\it left gyrotranslations} and
rotations of the space \cite[Sec.~4]{mybook03}.
The gyrotriangle inequality \eqref{rif01} reduces to a corresponding
gyrotriangle equality,
\begin{equation} \label{rif01s}
d_\op(\ab,\bb) = d_\op(\ab,\pb) \op d_\op(\pb,\bb)
\end{equation}
if and only if point $\pb$ lies between points $\ab$ and $\bb$, that is,
point $\pb$ lies on the gyrosegment $\ab\bb$, as shown in
Fig.~\ref{fig163b2enm} for points $A$, $B$, and $P$.
Accordingly, the gyrodistance function is gyroadditive on gyrolines, as
formulated in \eqref{rif01s}, and
illustrated graphiclly in Fig.~\ref{fig163b2enm}.

Furthermore, the Einstein gyrodistance function \eqref{tkcflsen} in any
$n$-dimensional Einstein gyrovector space $(\Rcn,\op,\od)$, $\Rcn\subset\Rn$
being the $n$-dimensional $c$-ball,
possesses a familiar Riemannian line element. It
gives rise to the Riemannian line element $ds_e^2$ of the Einstein gyrovector space
with its {\it gyrometric} \eqref{tkcflsen},
\begin{equation} \label{eqdelta01fen}
\begin{split}
ds_e^2 &=  \|(\xb+d\xb) \om \xb\|^2 \\[4pt]
&= \frac{c^2} {c^2-\xb^2}d\xb^2 +
  \frac{c^2} {(c^2-\xb^2)^2}(\xb\ccdot d\xb)^2
%\\[4pt] &=
%\frac{c^2 d\xb^2 - (\xb\timess d\xb)^2}{(c^2-\xb^2)^2}
\end{split}
\end{equation}
where $d\xb^2 = d\xb\ccdot d\xb$,
as shown in \cite[Theorem 7.6]{mybook03}.

Remarkably, the Riemannian line element $ds_e^2$ in \eqref{eqdelta01fen} turns out
to be the well-known Riemannian line element that the
Italian mathematician Eugenio Beltrami introduced in 1868
in order to study hyperbolic geometry by a Euclidean disc model,
now known as the Beltrami-Klein disc
\cite[p.~220]{mccleary94}\cite{barrett00}.
An English translation of his historically significant 1868 essay
on the interpretation of non-Euclidean geometry is found in
\cite{stillwell96}. The significance of Beltrami's 1868 essay
rests on the generally known fact that
it was the first to offer a concrete interpretation of hyperbolic geometry
by interpreting ``straight lines'' as geodesics on a surface of a constant
negative curvature.
Beltrami, thus, constructed a Euclidean disc model of the
hyperbolic plane \cite{mccleary94} \cite{stillwell96}, which now bears his name
along with the name of Klein.

We have thus found that the Beltrami-Klein
ball model of hyperbolic geometry is regulated algebraically by Einstein gyrovector spaces
with their gyrodistance function \eqref{tkcflsen}
and Riemannian line element \eqref{eqdelta01fen},
just as the standard model of Euclidean geometry is regulated algebraically by vector spaces
with their Euclidean distance function
and Riemannian line element $ds^2 = d\xb^2$.

%%%%%%%%%%%%%%%%%%%%%%%%%%%%%%%%%%%%%%%%%%%%%%%%%%%%%%%%%%%%%%%%%%%%
 %%%%%%%%%%%%%%%%%%%%%%%%%%%%%%%%%%%%%%%%%%%%%%%%%%%%%%%%%%%%%%%%%%%%%%
%%%%  Figs 174m and 175m %%%%%%%%%%%%%%%%%%%%%%%%%%%%%%%%%%%%%%%%%%%%%
%%%%%%%%%%%%%%%%%%%%%%%%%%%%%%%%%%%%%%%%%%%%%%%%%%%%%%%%%%%%%%%%%%%%%%
%%%%%%  fig127m  %%%%%% The Hyp. Pythagorean Theorem %%%%%%%%%%%%%%%%%
%%%%% A Right-Angled Hyp. Triangle and left Gyrotransl.% fig127m%%%%%%
%\begin{figure}[ht]
\begin{figure}[t]  % try to put this figure on the top of the page
              % [h] tries to place the figure here
              % [b] tries to place the figure on the bottom of the page
              % [t] tries to place the figure on the top of the page
              % [P] tries to place the figure floatingly on the page
 \sidebyside {       % center two figures
%%%%%%%%%%%%%%%%%%%%%%%%%%%%%%%%%%%%%%%%%%%%%%%%%%%%%%%%%%%%%%%%%%%%
% The [] causes centering the LaTex box in the psfrag comand
% The [] can be deleted if not needed.
%%%%%%%%%%%%%%%%%%%%%%%%%%%%%%%%%%%%%%%%%%%%%%%%%%%%  Left Figure %%
 \psfrag{pa}  {$A,\hspace{0.1cm}t=0$}
 \psfrag{pb}  {$B,\hspace{0.1cm}t=1$}
\psfrag{formula01}[]{${\rm The~Einstein~Gyroline} L_{^{AB}}$}
\psfrag{formula02}[]{${\rm through~the~points}~ A~{\rm and}~B$}
\psfrag{formula03}[]{$\boxed{A\op (\om  A\op  B)\od  t}$}
\psfrag{formula04}[]{$ -\infty < t < \infty $}
 \includegraphics[width=0.45\textwidth]{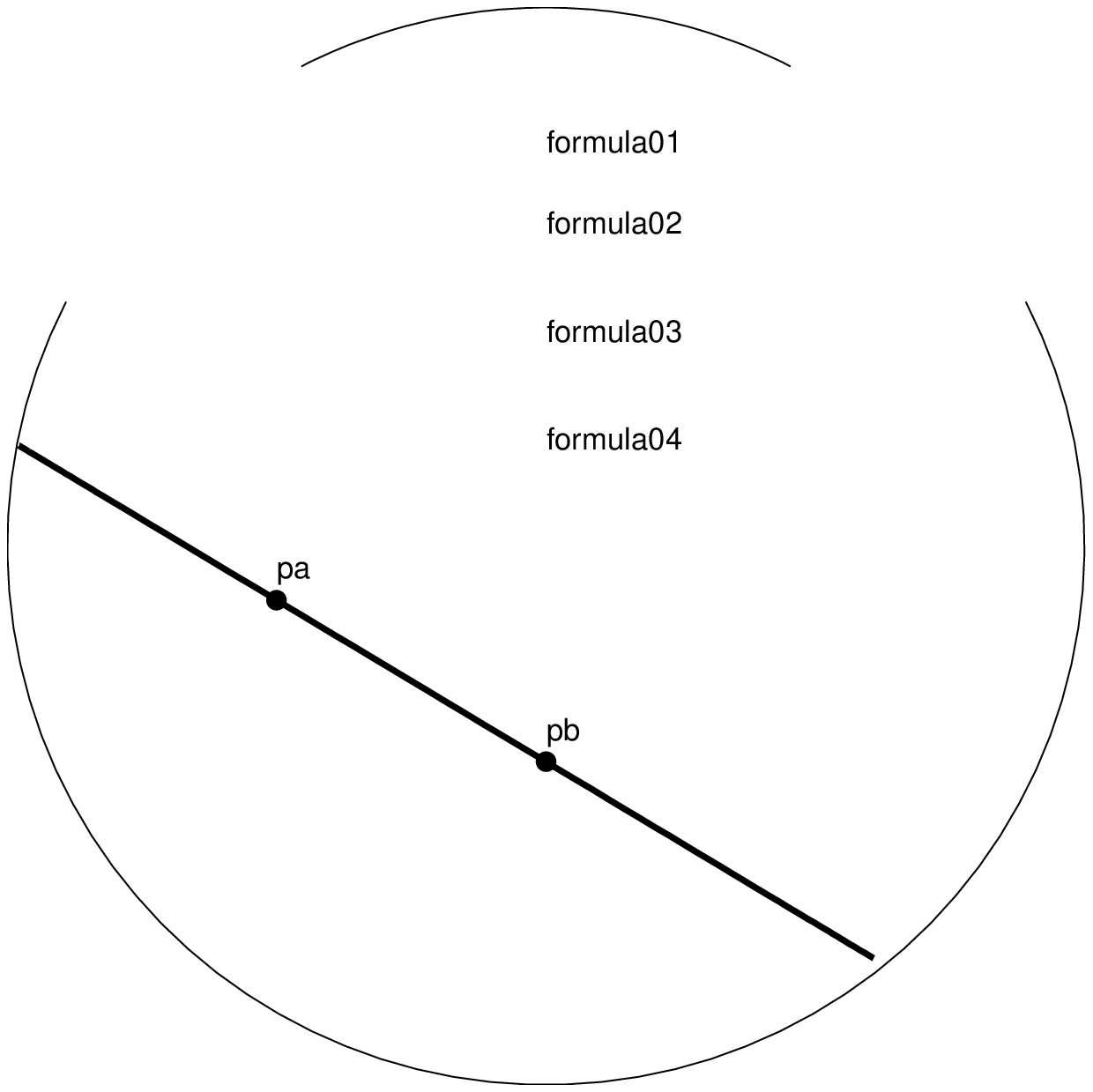}
\caption[The Einstein Gyroline]{
The unique gyroline $L_{^{AB}}$ in an Einstain gyrovector space $(\Rsn,\op ,\od )$
through two given points $A$ and $B$.
The case of the Einstain gyrovector plane, when $\Rsn=\Rstwou$ is the
real open unit disc, is shown.
\label{fig174b2enm}}}
%%%%%%%%%%%%%%%%%%%%%%%%%%%%%%%%%%%%%%%%%%%%%%%%%%%%  Rite Figure %%
  {
%\psfrag
\psfrag{a}[]{$A$}
\psfrag{b}[]{$\!\! B$}
\psfrag{mab}{$M_{^{AB}}^{\phantom{O}} \!=\! \half\od (A\sqp \! B)$}
\psfrag{pab}{$P$}
\psfrag{formula00}[]{$d_{\om }(A,P) \op  d_{\om }(P,B)=d_{\om }(A,B)$}
\psfrag{formula01}[]{$\boxed{A\op (\om  A\op  B)\od  t}$}
\psfrag{formula02}[]{$0\le t \le1$}

%\psfrag
%\includegraphics[width=0.45\textwidth]{/home/ungar/dir_amy/dir_papers/dir_mybook01/dir_figs/fig163b2en.eps}
 \includegraphics[width=0.45\textwidth]{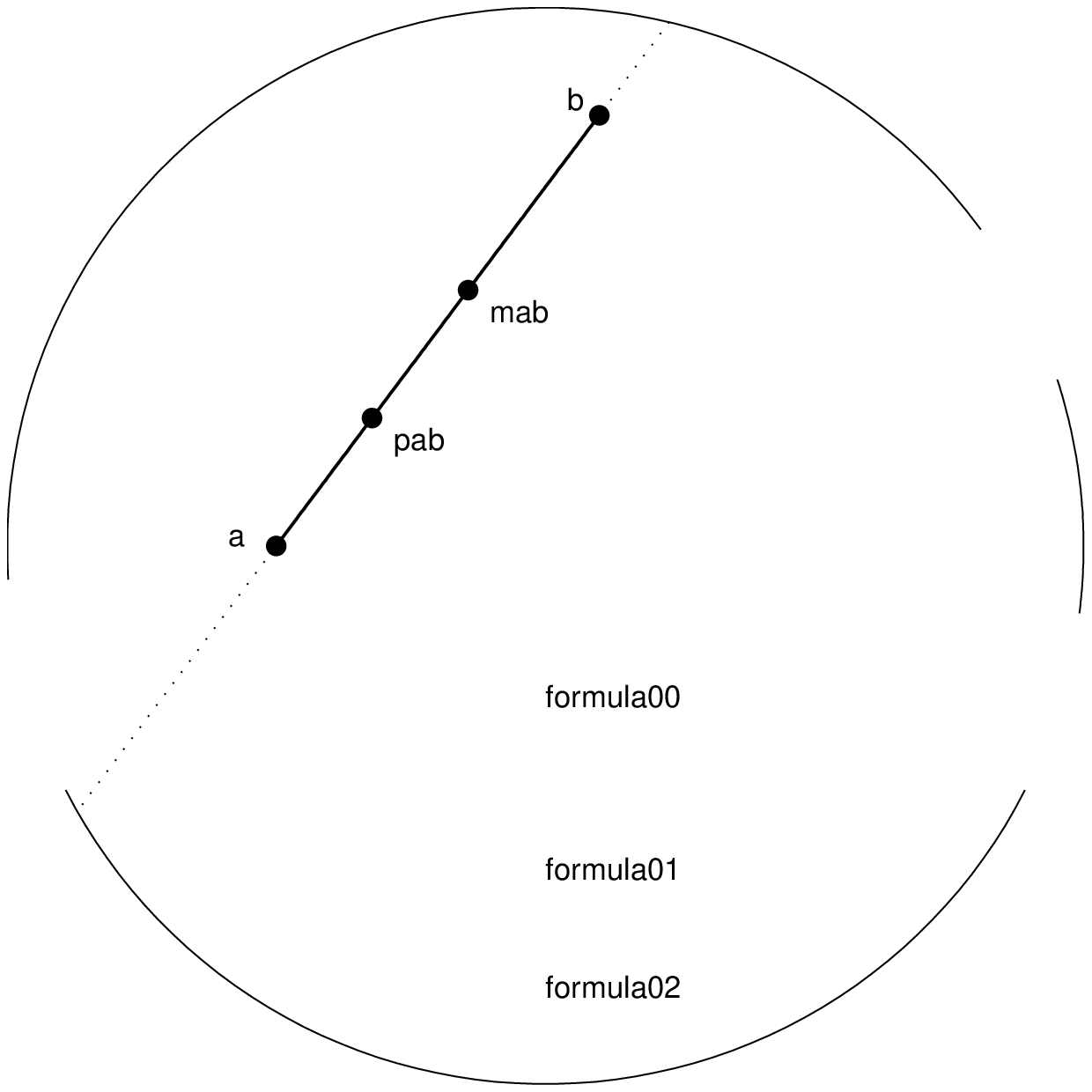}
\caption{
The gyrosegment $AB$ that links
the points $A$ and $B$ in $(\Rsn,\op ,\od )$, with one of its generic points
$P$ and its gyromidpoint $M_{^{AB}}$. The point $P$ lies between $A$ and $B$
and, hence, obeys the gyrotriangle equality, \eqref{rif01s}.
\label{fig163b2enm}} }
\end{figure}
%%%%%%%%%%%%%%%%%%%%%%%%%%%%%%%%%%%%%%%%%%%%%%%%%%%%%%%%%%%%%%%%%%%%%%

 %Figure - Einstein gyroline and midpoint and
                       %           Gyrotriangle Equality.
                     % Figs. 1 and 2
                     % Fig.~\ref{fig174b2enm}, Fig.~\ref{fig163b2enm}
%%%%%%%%%%%%%%%%%%%%%%%%%%%%%%%%%%%%%%%%%%%%%%%%%%%%%%%%%%%%%%%%%%%%

In full analogy with Euclidean geometry,
the unique Einstein gyroline $\LAB$, Fig.~\ref{fig174b2enm},
that passes through two given
points $A$ and $B$ in an Einstein gyrovector space $\VS=(\VS,\op ,\od)$
is given by the equation
%%%%%%%%%%%%%%%%%%%%%%%%%%%%%%%%%%%%%%%%%%%%%%%%%%%%%%%%%%%%%%%%%%%%
 \begin{equation} \label{eqrfgh01en}
\LAB = A\op (\om A \op  B)\od  t
 \end{equation}
%%%%%%%%%%%%%%%%%%%%%%%%%%%%%%%%%%%%%%%%%%%%%%%%%%%%%%%%%%%%%%%%%%%%

Einstein gyrolines are chords of the ball, which turn out to be the
familiar geodesics of the Beltrami-Klein ball model of hyperbolic geometry
\cite{mccleary94}.
Accordingly, Einstein gyrosegments are Euclidean segments, as shown in
Fig.~\ref{fig163b2enm}.

The gyromidpoint $\MAB$ of gyrosegment $AB$ is the
unique point of the gyrosegment that satisfies the equation
$d_\op(\MAB,A) = d_\op(\MAB,B)$. It is given by each of the following equations
\cite[Theorem 3.33]{mybook04}, Fig.~\ref{fig163b2enm},
\begin{equation} \label{eqgyromid}
\MAB = A\op (\om A \op  B)\od \half
     = \frac{\gamma_{_A}A+\gamma_{_B}B}{\gamma_{_A}+\gamma_{_B}} =
\half \od (A\sqp B)
\end{equation}
in full analogy with Euclidean midpoints.

In Euclidean geometry a parallelogram is a quadrilateral the two diagonals
of which intersect at their midpoints.
In full analogy, in hyperbolic geometry a gyroparallelogram is a gyroquadrilateral
the two gyrodiagonals of which intersect at their gyromidpoints.
Accordingly, if $A$, $B$ and $C$ are any three non-gyrocollinear points
(that is, they do not lie on a gyroline) in an Einstein gyrovector space,
and if a fourth point $D$ is given by the {\it gyroparallelogram condition}
\begin{equation} \label{hkmdnf1}
D = (B\sqp C)\om A
\end{equation}
then the gyroquadrilateral $ABDC$ is a gyroparallelogram, shown in
Fig.~\ref{fig190k1m}.

%%%%%%%%%%%%%%%%%%%%%%%%%%%%%%%%%%%%%%%%%%%%%%%%%%%%%%%%%%%%%%%%%%%%
  
%%%%%%%%%%%%%%%%%%%%%%%%%%%%%%%%%%%%%%%%%%%%%%%%%%%%%%%%%%%%%%%%%%%%%%
%%%%% The Einstein gyroparallelogram                %%%%%%%%%%%%%%%%%%
%\begin{figure}[htbp]
\begin{figure}[t]  % try to put this figure on the top of the page
              % [h] tries to place the figure here
              % [b] tries to place the figure on the bottom of the page
              % [t] tries to place the figure on the top of the page
              % [P] tries to place the figure floatingly on the page
 \centering         % center the figure
\psfrag{----chets}[]{\lower-1.0ex \hbox {\footnotesize{$\blacktriangleright$}}}
\psfrag{mab}[]{$\hspace{0.66cm}\MABDC$}
\psfrag{O}[]{$\phantom{O}$}
\psfrag{a}[]{$A$}
\psfrag{b}[]{$B$}
\psfrag{c}[]{$C$}
\psfrag{d}[]{$D$}
\psfrag{formula00}{${\rm The~Gyroparallelogram}$} 
\psfrag{formula00a}{${\rm Condition:}~~~~D = (B\sqp C)\om A$} 
\psfrag{formula01}{$\MAD = \frac{\gamma_{_A}A+\gamma_{_D}D}{\gamma_{_A}+\gamma_{_D}}
 = \half(A\sqp D)$}
\psfrag{formula02}{$\MBC = \frac{\gamma_{_B}B+\gamma_{_C}C}{\gamma_{_B}+\gamma_{_C}}
 = \half(A\sqp C)$}
\psfrag{formula03}{$\MABDC = \frac{\gmA A +\gmB B +\gmC C +\gmD D}{\gmA+\gmB+\gmC+\gmD}$}
 \psfrag{formula03a}{$\MABDC=\MAD=\MBC$}
 \psfrag{formula04}{$\om C \op D =\gyr[ C ,\om B ]\gyr[ B ,\om A ](\om A \op B )$}
\psfrag{formula05}{$\om B \op D =\gyr[ B ,\om C ]\gyr[ C ,\om A ](\om A \op C )$}
\psfrag{fig190j}[]{$\boxed{(\om A \op B )\sqp(\om A \op C )=\om A \op D }$}
\psfrag{formula06}[]{$\boxed{\ub\sqp\vb=\wb}$}
\psfrag{text1}[]{$\hspace{0.20cm}\ub=\om A\op B$}
\psfrag{text2}[]{$\vb=\om A\op C$}
\psfrag{text3}[]{$\wb=\om A\op D$}
 \includegraphics[width=9cm]{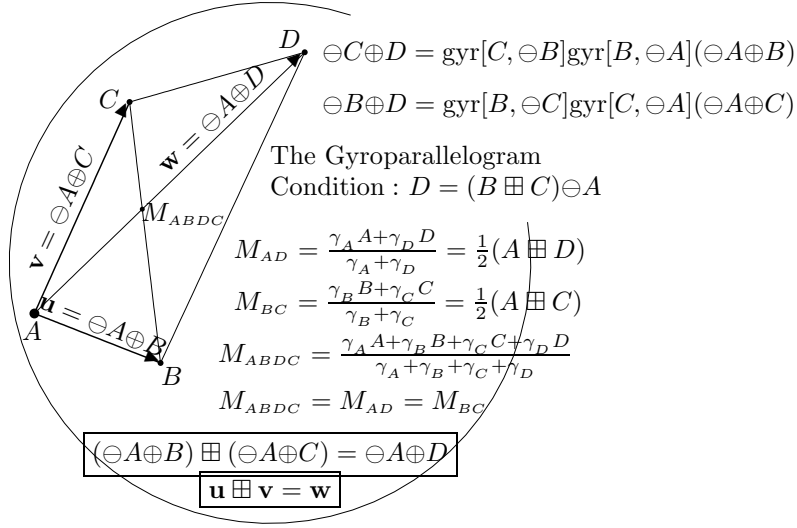}
\caption{
The Einstein gyroparallelogram and its addition law. 
The gyrodiagonals $AD$ and $BC$ of gyroparallelogram $ABDC$
intersect each other at their gyromidpoints.
Detailed studies of the gyroparallelogram and its extension to higher dimensional
gyroparallelepipeds are presented in \cite{mybook02,mybook03}.
The gyroparallelogram addition law plays an important role in our
gyrovector space approach to hyperbolic geometry, studied in \cite{mybook03,mybook04}.
\label{fig190k1m}}
\end{figure}
%%%%%%%%%%%%%%%%%%%%%%%%%%%%%%%%%%%%%%%%%%%%%%%%%%%%%%%%%%%%%%%%%%%%%%

             %Figure - Einstein gyroparallelogram
                     % Fig. 3
                     %                         Fig.~\ref{fig190k1m}
%%%%%%%%%%%%%%%%%%%%%%%%%%%%%%%%%%%%%%%%%%%%%%%%%%%%%%%%%%%%%%%%%%%%

Indeed, the two gyrodiagonals of gyroquadrilateral $ABDC$ are the gyrosegments
$AD$ and $BC$, shown in Fig.~\ref{fig190k1m},
the gyromidpoints of which coincide, that is,
\begin{equation} \label{hkmdnf2}
\half\od(A\sqp D) =  \half\od(B\sqp C)
\end{equation}
The result in \eqref{hkmdnf2} follows
from the gyroparallelogram condition \eqref{hkmdnf1},
the implication in \eqref{hufns} and a
gyromidpoint identity in \eqref{eqgyromid}.

In his 1905 paper that founded the special theory of relativity \cite{einstein05},
Einstein noted about his addition law of relativistically admissible velocities:
\begin{quotation}
``Das Gesetz vom Parallelogramm der Geschwindigkeiten gilt also nach
unserer Theorie nur in erster Ann\"aherung.''
\begin{flushright}
A.~Einstein \cite{einstein05}
\end{flushright}
\end{quotation}
[English translation: Thus the law of velocity parallelogram is valid according to our
theory only to a first approximation.]

Indeed, Einstein velocity addition, $\op$, is noncommutative and does not give rise
to an exact ``velocity parallelogram'' in Euclidean geometry.
However, as we see in Fig.~\ref{fig190k1m}, Einstein velocity {\it coaddition}, $\sqp$,
which is commutative, does give rise to an exact
``velocity gyroparallelogram'' in hyperbolic geometry.

The breakdown of commutativity in Einstein velocity addition law seemed
undesirable to the famous mathematician \'Emile Borel.
Borel's resulting attempt to ``repair'' the seemingly ``defective'' Einstein velocity
addition in the years following 1912 is described by Walter in
\cite[p.~117]{walter99b}.
Here, however, we see that there is no need to repair Einstein velocity
addition law for being noncommutative since, despite of being noncommutative,
it gives rise to the gyroparallelogram law of gyrovector addition, which
turns out to be commutative.
The compatibility of our gyroparallelogram addition law of Einsteinian velocities
with cosmological observations of stellar aberration will be discussed in
Sec.~\ref{secc8}.
The extension of the gyroparallelogram addition law of $n=2$ summands to a
corresponding gyroparallelepiped addition law of $n>2$ summands is
presented in \cite[Theorem 10.6]{mybook03}.

% SECTION NUMBER 5
\section{Vectors and Gyrovectors}\label{secc5}

%%%%%%%%%%%%%%%%%%%%%%%%%%%%%%%%%%%%%%%%%%%%%%%%%%%%%%%%%%%%%%%%%%%%
%%%  Double Figs.  %%%%%%%%%%%%%%%%%%%%%%%%%%%%%%%%%%%%%%%%%%%%%%%%%%%%%%%%%%
%%%%%%%%%%%%%%%%%%%%%%%%%%%%%%%%%%%%%%%%%%%%%%%%%%%%%%%%%%%%%%%%%%%%%%
%\begin{figure}[ht]
\begin{figure}[t]  % try to put this figure on the top of the page
              % [h] tries to place the figure here
              % [b] tries to place the figure on the bottom of the page
              % [t] tries to place the figure on the top of the page
              % [P] tries to place the figure floatingly on the page
 \sidebyside {       % center two figures
\psfrag{A}[]{$P$}
\psfrag{B}[]{$Q$}
\psfrag{Ap}[]{$R$}
\psfrag{Bp}[]{$S$}
\psfrag{formula00}[]{$-P+Q=-R+S$}
\psfrag{formula01}[]{$d(P,Q)=\|-P+Q\|$}
\psfrag{---chets1}[]{\lower-1.2ex \hbox {\footnotesize{$\blacktriangleright$}}}
\psfrag{---chets2}[]{\lower-1.2ex \hbox {\footnotesize{$\blacktriangleright$}}}
 \includegraphics[width=0.5\textwidth]{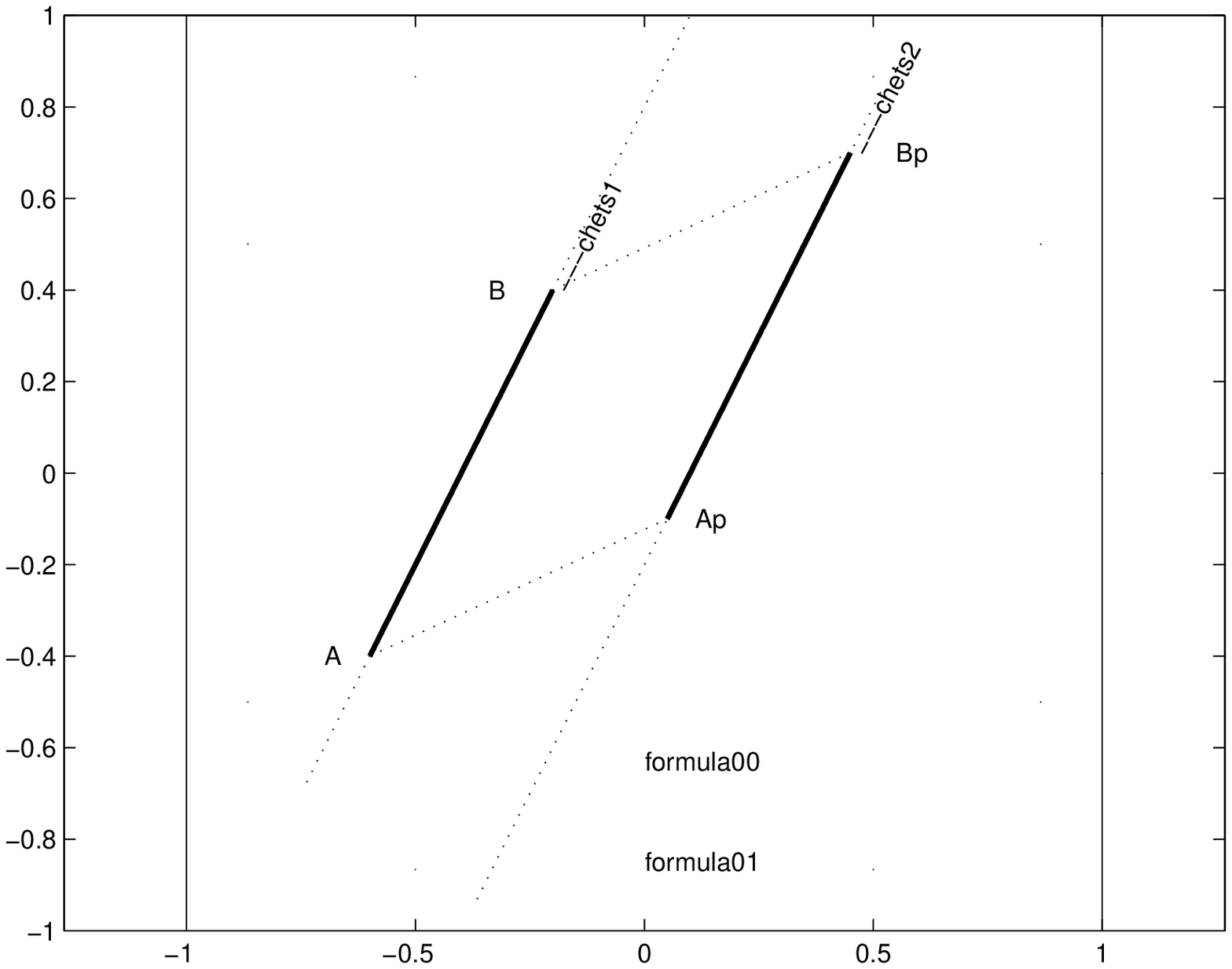}
\caption{
Two equivalent vectors in a Euclidean vector plane $(\Rtwo,+,\ccdot)$.
The two vectors are parallel and have equal values and, hence, equal lengths.
\label{fig221a4m}}}
  {
\psfrag{A}[]{$P$}
\psfrag{B}[]{$Q$}
\psfrag{Ap}[]{$R$}
\psfrag{Bp}[]{$S$}
\psfrag{formula00}[]{$\om P\op Q=\om R\op S$}
\psfrag{formula01}[]{$d(P,Q)=\|\om P\op Q\|$}
\psfrag{---chets1}[]{\lower-1.2ex \hbox {\footnotesize{$\blacktriangleright$}}}
%\psfrag{---chets2}[]{\lower-1.2ex \hbox {\footnotesize{$\blacktriangleright$}}}
 \psfrag{---chets2}[]{\lower-1.0ex \hbox {\footnotesize{$\blacktriangleright$}}}
 \includegraphics[width=0.5\textwidth]{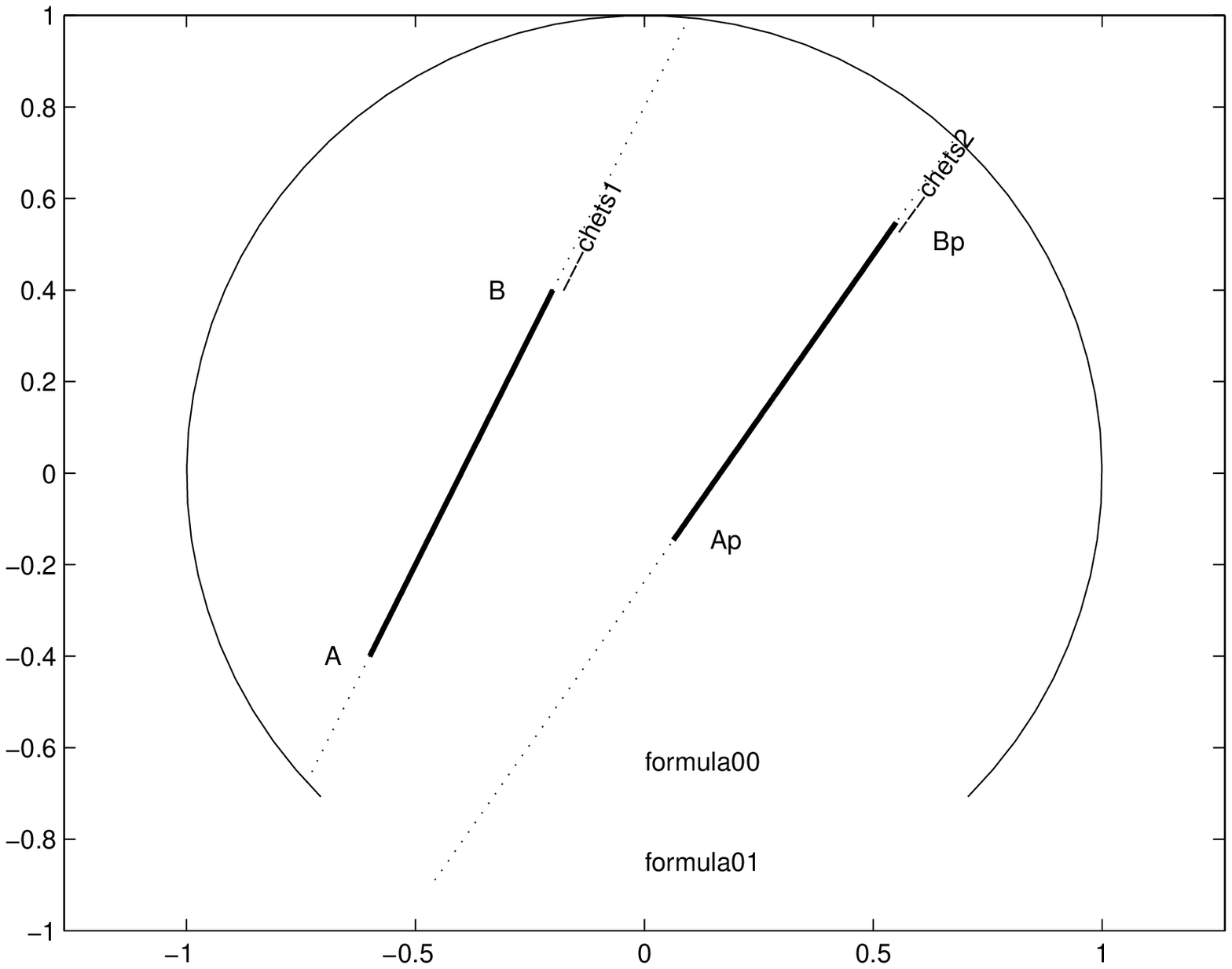}
\caption[]{
Two equivalent gyrovectors in an Einstein gyrovector plane $(\Rctwo,\op,\od)$.
The two gyrovectors have equal values and, hence, equal gyrolengths.
\label{fig221b4m}}}
\end{figure}
%%%%%%%%%%%%%%%%%%%%%%%%%%%%%%%%%%%%%%%%%%%%%%%%%%%%%%%%%%%%%%%%%%%%%%
 % Vector and Gyrovector Translation in Euc. and Mb.
%  Figs 4 and 5
%  Fig.~\ref{fig221a4m}        Vector Translation in a Euc. plane
%  Fig.~\ref{fig221b4m}        Gyrovector Translation in an Einstein plane
%%%%%%%%%%%%%%%%%%%%%%%%%%%%%%%%%%%%%%%%%%%%%%%%%%%%%%%%%%%%%%%%%%%%

Elements of a real inner product space $\vi=(\vi,+,\ccdot)$, called
points and denoted by capital italic letters, $A,B,P,Q,$ etc,
give rise to vectors in $\vi$, denoted by bold roman lowercase letters $\ub,\vb,$ etc.
Any two ordered points $P,Q\inn\vi$ give rise to a unique rooted vector
$\vb\inn\vi$, rooted at the point $P$. It has a tail at the point $P$
and a head at the point $Q$, and it has the value $-P+Q$,
\begin{equation} \label{eq2rhkd01}
\vb =-P+Q
\end{equation}
The length of the rooted vector $\vb = -P+Q$ is the distance between the points
$P$ and $Q$, given by the equation
\begin{equation} \label{eq2rhkd02}
\|\vb\| = \|-P+Q\|
\end{equation}

Two rooted vectors $-P+Q$ and $-R+S$ are equivalent if they have the same
value, that is,
\begin{equation} \label{fksnc}
-P+Q ~~\thicksim~~-R+S \hspace{1.0cm} {\rm if~and~only~if} \hspace{1.0cm} -P+Q=-R+S
\end{equation}
The relation $\thicksim$ in \eqref{fksnc} between rooted vectors is
reflexive, symmetric and transitive, so that it is an equivalence relations that gives
rise to equivalence classes of rooted vectors.
Two equivalent rooted vectors in a Euclidean vector plane are shown in Fig.~\ref{fig221a4m}.
To liberate rooted vectors from
their roots we define
a {\it vector} to be an equivalence class of rooted vectors.
The vector $-P+Q$ is thus a representative of all rooted vectors with value $-P+Q$.

A point $P\inn\vi$ is identified with the vector $-O+P$,
$O$ being the arbitrarily selected origin of the space $\vi$.
Hence, the algebra of vectors can be applied to points as well.
Naturally, geometric and physical properties regulated by a vector space are
independent of the choice of the origin.

Let $A,B,C\in\vi$ be three non-collinear points, and let
%%%%%%%%%%%%%%%%%%%%%%%%%%%%%%%%%%%%%%%%%%%%%%%%%%%%%%%%%%%%%%%%%%%%
\begin{equation} \label{ksfr1}
\begin{split}
\ub &= -A+B \\
\vb &= -A+C
\end{split}
\end{equation}
%%%%%%%%%%%%%%%%%%%%%%%%%%%%%%%%%%%%%%%%%%%%%%%%%%%%%%%%%%%%%%%%%%%%
be two vectors in $\vi$ that, without loss of generality,
possess the same tail, $A$. Furthermore, let $D$ be a point of $\vi$ given by the
parallelogram condition
\begin{equation} \label{ksfr2}
D = B+C-A
\end{equation}
Then, the quadrilateral $ABDC$ is a parallelogram in Euclidean geometry in the sense
that its two diagonals, $AD$ and $BC$, intersect at their midpoints, that is,
\begin{equation} \label{ksfr3}
\half(A+D) = \half(B+C)
\end{equation}
Then, the vector addition of the vectors $\ub$ and $\vb$
that generate the parallelogram $ABDC$
is the vector $\wb$
given by the well-known parallelogram addition law,
\begin{equation} \label{ksfr4}
\wb = -A+D
\end{equation}
where $\wb$ is the vector formed by the diagonal $AD$ of the parallelogram $ABDC$.

Vectors in the space $\vi$
are, thus, equivalence classes of ordered pairs of points that add according
to the parallelogram law.

Gyrovectors emerge in Einstein gyrovector spaces $\VS=(\VS,\op,\od)$ in a way fully analogous
to the way vectors emerge in the spaces $\vi$.
Elements of $\VS$, called
points and denoted by capital italic letters, $A,B,P,Q,$ etc,
give rise to gyrovectors in $\VS$, denoted by bold roman lowercase letters $\ub,\vb,$ etc.
Any two ordered points $P,Q\inn\VS$ give rise to a unique rooted gyrovector
$\vb\inn\VS$, rooted at the point $P$. It has a tail at the point $P$
and a head at the point $Q$, and it has the value $\om P\op Q$,
\begin{equation} \label{eq2rhkd01g}
\vb =\om P\op Q
\end{equation}
The gyrolength of the rooted gyrovector $\vb = \om P\op Q$ is the gyrodistance between the points
$P$ and $Q$, given by the equation
\begin{equation} \label{eq2rhkd02g}
\|\vb\| = \|\om P\op Q\|
\end{equation}

Two rooted gyrovectors $\om P\op Q$ and $\om R\op S$ are equivalent
if they have the same value, that is,
\begin{equation} \label{fksnd}
\om P\op Q ~~\thicksim~~ \om R\op S \hspace{1.0cm} {\rm if~and~only~if} \hspace{1.0cm}
\om P\op Q= \om R \op S
\end{equation}
The relation $\thicksim$ in \eqref{fksnd} between rooted gyrovectors is
reflexive, symmetric and transitive, so that it is an equivalence relation that gives
rise to equivalence classes of rooted gyrovectors.
Two equivalent rooted gyrovectors in an Einstein gyrovector plane are shown in
Fig.~\ref{fig221b4m}.
To liberate rooted gyrovectors from
their roots we define
a {\it gyrovector} to be an equivalence class of rooted gyrovectors.
The gyrovector $\om P\op Q$ is thus a representative of all rooted gyrovectors
with value $\om P\op Q$.

A point $P$ of a gyrovector space $(\VS,\op,\od)$ is identified with the gyrovector $\om O\op P$,
$O$ being the arbitrarily selected origin of the space $\VS$.
Hence, the algebra of gyrovectors can be applied to points as well.
Naturally, geometric and physical properties regulated by a gyrovector space are
independent of the choice of the origin.

Let $A,B,C\in\VS$ be three non-gyrocollinear points of an Einstein gyrovector space
$(\VS,\op,\od)$, and let
%%%%%%%%%%%%%%%%%%%%%%%%%%%%%%%%%%%%%%%%%%%%%%%%%%%%%%%%%%%%%%%%%%%%
\begin{equation} \label{ktfr1}
\begin{split}
\ub &= \om A\op B \\
\vb &= \om A\op C
\end{split}
\end{equation}
%%%%%%%%%%%%%%%%%%%%%%%%%%%%%%%%%%%%%%%%%%%%%%%%%%%%%%%%%%%%%%%%%%%%
be two gyrovectors in $\vi$ that, without loss of generality,
possess the same tail, $A$. furthermore, let $D$ be a point of $\VS$ given by the
gyroparallelogram condition
\begin{equation} \label{ktfr2}
D = (B\sqp C)\om A
\end{equation}
Then, the gyroquadrilateral $ABDC$ is a gyroparallelogram in the
Beltrami-Klein ball model of hyperbolic geometry in the sense
that its two gyrodiagonals, $AD$ and $BC$, intersect at their gyromidpoints, that is,
\begin{equation} \label{ktfr3}
\half(A\sqp D) = \half(B\sqp C)
\end{equation}
as explained in Sec.~\ref{secc4} and illustrated in Fig.~\ref{fig190k1m}.
Then, the gyrovector addition of the gyrovectors $\ub$ and $\vb$
that generate the gyroparallelogram $ABDC$ is the gyrovector $\wb$
given by the gyroparallelogram addition law, Fig.~\ref{fig190k1m},
\begin{equation} \label{ktfr4}
\wb = \om A \op D = \ub \sqp \vb
\end{equation}
where $\wb$ is the gyrovector formed by the gyrodiagonal $AD$ of the gyroparallelogram $ABDC$.

Gyrovectors in the ball $\VS$
are, thus, equivalence classes of ordered pairs of points that add according
to the gyroparallelogram law.

% SECTION NUMBER 6
\section{Gyrotrigonometry in Einstein Gyrovector Spaces}\label{secc6}

%%%%%%%%%%%%%%%%%%%%%%%%%%%%%%%%%%%%%%%%%%%%%%%%%%%%%%%%%%%%%%%%%%%%
  
%%%%%%%%%%%%%%%%%%%%%%%%%%%%%%%%%%%%%%%%%%%%%%%%%%%%%%%%%%%%%%%%%%%%%%
%%%%% The Einstein gyroparallelogram                %%%%%%%%%%%%%%%%%%
%\begin{figure}[htbp]
\begin{figure}[t]  % try to put this figure on the top of the page
              % [h] tries to place the figure here
              % [b] tries to place the figure on the bottom of the page
              % [t] tries to place the figure on the top of the page
              % [P] tries to place the figure floatingly on the page
 \centering         % center the figure
\psfrag{O}[]{$\phantom{O}$}
\psfrag{A}[]{$A$}
%\psfrag{B}[]{$\hspace{0.2cm}B$}
\psfrag{B}[]{\lower-1.2ex \hbox {$\hspace{0.2cm}B$}}
\psfrag{C}[]{$C$}
\psfrag{pa}{$a$}
\psfrag{pb}{$b$}
\psfrag{pc}{$c$}
\psfrag{text1}[]{$\cb=\om A \op B$}
\psfrag{text2}[]{$\ab=\om B \op C$}
\psfrag{text2a}[]{$\ab^\prime=\om C\op B=\om\gyr[\om C,B]\ab$}
\psfrag{text3}[]{$\bb=\om C \op A$}
\psfrag{al}[]{$\alpha$}
\psfrag{be}[]{$\beta$}
\psfrag{ga}[]{$\gamma$}
\psfrag{formula01}{$a=\|\ab\|=\|\om B \op C\| = |BC|$}
\psfrag{formula02}{$b=\|\bb\|=\|\om C \op A\| = |CA|$}
\psfrag{formula03}{$c=\|\cb\|=\|\om A \op B\| = |AB|$}
\psfrag{formula03a}{$a_s=\frac{a}{s},\hspace{0.4cm}b_s=\frac{b}{s},\hspace{0.4cm}c_s=\frac{c}{s}$}
\psfrag{formula04}[]{$\cos\alpha=\frac{\om A\op B}{\|\om A\op B\|}
\ccdot \frac{\om A\op C}{\|\om A\op C\|}$}
\psfrag{formula05}[]{$\cos\beta =\frac{\om B\op A}{\|\om B\op A\|}
\ccdot \frac{\om B\op C}{\|\om B\op C\|}$}
\psfrag{formula06}[]{$\cos\gamma=\frac{\om C\op A}{\|\om C\op A\|}
\ccdot \frac{\om C\op B}{\|\om C\op B\|}$}
\psfrag{fig171ec4}{$\delta=\pi-(\alpha+\beta+\gamma)>0$}
%%%%%%%%%%%%%%%%%%%%%%%%%%%%%%%%%%%%%%%%%%%%%%%%%%%%%%%%%%%%%%%%%%%%%%
\psfrag{----chets1}[]{\lower-1.2ex \hbox {$\blacktriangleright$}}
\psfrag{----chets2}[]{\lower-1.26ex \hbox {$\blacktriangleright$}}
\psfrag{----chets3}[]{\lower-1.2ex \hbox {$\blacktriangleright$}}
%%%%%%%%%%%%%%%%%%%%%%%%%%%%%%%%%%%%%%%%%%%%%%%%%%%%%%%%%%%%%%%%%%%%%%
%\includegraphics[width=10cm]{/home/ungar/dir_amy/dir_papers/dir_mybook01/dir_figs/fig171ec4.eps}
 \includegraphics[width=10cm]{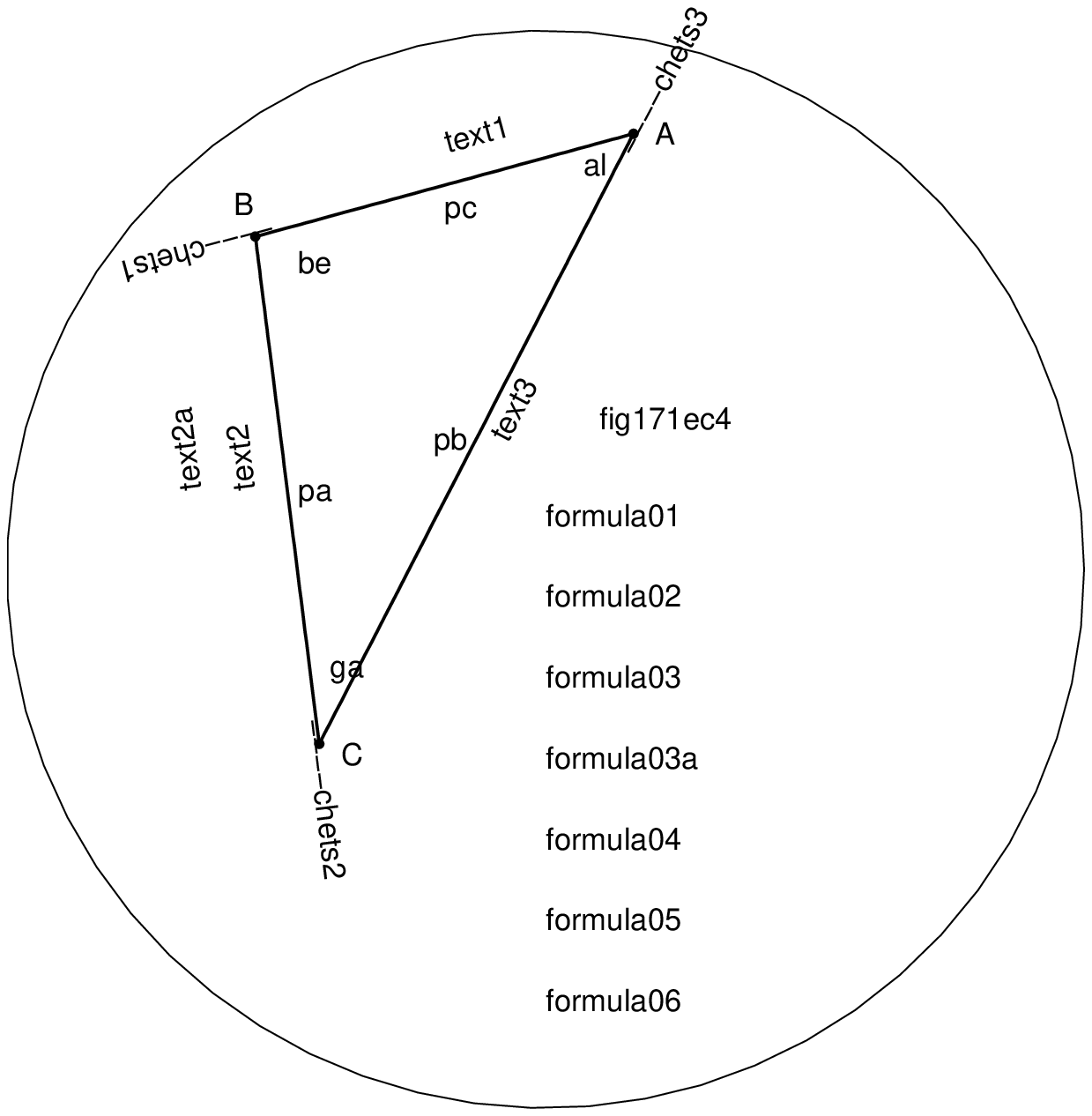}
\caption{
The gyrotriangle, and its standard notation, in an Einstein gyrovector space
$(\VS,\op,\od)$.
\label{fig171ec4m}}
\end{figure}
%%%%%%%%%%%%%%%%%%%%%%%%%%%%%%%%%%%%%%%%%%%%%%%%%%%%%%%%%%%%%%%%%%%%%%
            %Figure - Einstein gyrotriangle
                     % Fig. 6
                     %                         Fig.~\ref{fig171ec4m}
%%%%%%%%%%%%%%%%%%%%%%%%%%%%%%%%%%%%%%%%%%%%%%%%%%%%%%%%%%%%%%%%%%%%

Let $A,B$ and $C$ be three non-gyrocollinear points
in an Einstein gyrovector space $(\VS,\op,\od)$, and
let $\om A \op B$ and $\om A \op C$ be the resulting two nonzero gyrovectors
with the common tail $A$ and included gyroangle $\alpha$,
as shown in Fig.~\ref{fig171ec4m}.
The gyrolength of gyrovector $\om A \op B$, $\|\om A \op B \|$, is nonzero,
and its associated gyrovector
\begin{equation} \label{eq5thjks}
\frac{\om A \op B }{\|\om A \op B \|}
\end{equation}
is called a unit gyrovector.
Guided by analogies with Euclidean trigonometry, the gyrocosine, cos, of the gyroangle
$\alpha$ included by the two gyrovectors that emanate from point $A$
is given by the equation
\begin{equation}\label{eqangle00}
\cos  \alpha=\frac{\om A \op B }{\|\om A \op B \|}\ccdot
\frac{\om A \op C }{\|\om A \op C \|}
\end{equation}
as illustrated in Fig.~\ref{fig171ec4m} for gyrotriangle gyroangles.
The gyroangle is invariant under the group of motions of any
Einstein gyrovector space, that is, under {\it left gyrotranslations} and
rotations of the space \cite[Theorem 8.6]{mybook03}.

Applying the gamma identity \eqref{eqgupv00} and the definition of the gyrocosine function
in \eqref{eqangle00} to the gyrovector sides of a gyrotriangle $ABC$, using the standard
gyrotriangle notation in Fig.~\ref{fig171ec4m}, we obtain the following
{\it law of gyrocosines}
\cite[Sec.~4.5]{mybook04},
\cite[Sec.~12.2]{mybook03},
%%%%%%%%%%%%%%%%%%%%%%%%%%%%%%%%%%%%%%%%%%%%%%%%%%%%%%%%%%%%%%%%%%%%
 \begin{equation} \label{jsf01}
 \begin{split}
 \gamma_a &= \gamma_b\gamma_c (1-b_s c_s \cos\alpha) \\
 \gamma_b &= \gamma_a\gamma_c (1-a_s c_s \cos\beta ) \\
 \gamma_c &= \gamma_a\gamma_b (1-a_s b_s \cos\gamma)
 \end{split}
 \end{equation}
%%%%%%%%%%%%%%%%%%%%%%%%%%%%%%%%%%%%%%%%%%%%%%%%%%%%%%%%%%%%%%%%%%%%

Like Euclidean triangles, the gyroangles of a gyrotriangle are
uniquely determined by its sides.
Solving the system \eqref{jsf01} of three identities for the
three unknowns $\cos\alpha$, $\cos\beta$ and $\cos\gamma$,
and employing \eqref{tksnw}, we obtain the following theorem.
%%%%%%%%%%%%%%%%%%%%%%%%%%%%%%%%%%%%%%%%%%%%%%%%%%%%%%%%%%%%%%%%%%%%
% THEOREM NUMBER 12.1
% THEOREM NUMBER 4.7
\begin{ttheorem}\label{thmgyrocos}
{\bf (The Law of Gyrocosines;
The $SSS$ to $AAA$ Conversion Law).}
Let $ABC$ be a gyrotriangle in an Einstein gyrovector space $(\VS,\op,\od)$.
Then, in the gyrotriangle notation in Fig.~\ref{fig171ec4m},

%%%%%%%%%%%%%%%%%%%%%%%%%%%%%%%%%%%%%%%%%%%%%%%%%%%%%%%%%%%%%%%%%%%%
\begin{equation} \label{jsf02}
\begin{split}
\cos\alpha &= \frac{-\gamma_a + \gamma_b \gamma_c}
                   {\gamma_b \gamma_c b_s c_s} 
= \frac{-\gamma_a + \gamma_b \gamma_c}
       {\sqrt{\gamma_b^2 -1} \sqrt{\gamma_c^2 -1}}
\\[6pt]
\cos\beta  &= \frac{-\gamma_b + \gamma_a \gamma_c}
                   {\gamma_a \gamma_c a_s c_s} 
= \frac{-\gamma_b + \gamma_a \gamma_c}
       {\sqrt{\gamma_a^2 -1} \sqrt{\gamma_c^2 -1}}
\\[6pt]
\cos\gamma &= \frac{-\gamma_c + \gamma_a \gamma_b}
                   {\gamma_a \gamma_b a_s b_s} 
= \frac{-\gamma_c + \gamma_a \gamma_b}
       {\sqrt{\gamma_a^2 -1} \sqrt{\gamma_b^2 -1}}
\end{split}
\end{equation}
%%%%%%%%%%%%%%%%%%%%%%%%%%%%%%%%%%%%%%%%%%%%%%%%%%%%%%%%%%%%%%%%%%%%
\end{ttheorem}
%%%%%%%%%%%%%%%%%%%%%%%%%%%%%%%%%%%%%%%%%%%%%%%%%%%%%%%%%%%%%%%%%%%%

The identities in \eqref{jsf02} form the
$SSS$ (Side-Side-Side) to $AAA$ (gyroAngle-gyroAngle-gyroAngle) conversion law in
Einstein gyrovector spaces.
This law is useful for calculating the gyroangles of a gyrotriangle in an
Einstein gyrovector space when its sides (that is, its side-gyrolengths) are known.

In full analogy with the trigonometry of triangles,
the {\it gyrosine}
of a gyrotriangle gyroangle $\alpha$ is nonnegative,
given by the equation
\begin{equation} \label{djbshe}
\sin\alpha = \sqrt{1-\cos^2\alpha}
\end{equation}
Hence, it follows from Theorem \ref{thmgyrocos} that the gyrosine of the
gyrotriangle gyroangles in that Theorem are given by
%%%%%%%%%%%%%%%%%%%%%%%%%%%%%%%%%%%%%%%%%%%%%%%%%%%%%%%%%%%%%%%%%%%%
\begin{equation} \label{fskbrz}
\begin{split}
\sin\alpha &= \frac{
\sqrt{1 + 2\gamma_a \gamma_b \gamma_c - \gamma_a^2-\gamma_b^2-\gamma_c^2}
}
{\sqrt{\gamma_b^2 -1} \sqrt{\gamma_c^2 -1}}
\\[4pt]
\sin\beta  &= \frac{
\sqrt{1 + 2\gamma_a \gamma_b \gamma_c - \gamma_a^2-\gamma_b^2-\gamma_c^2}
}
{\sqrt{\gamma_a^2 -1} \sqrt{\gamma_c^2 -1}}
\\[4pt]
\sin\gamma &= \frac{
\sqrt{1 + 2\gamma_a \gamma_b \gamma_c - \gamma_a^2-\gamma_b^2-\gamma_c^2}
}
{\sqrt{\gamma_a^2 -1} \sqrt{\gamma_b^2 -1}}
\end{split}
\end{equation}
% MATHEMATICA stam089
%%%%%%%%%%%%%%%%%%%%%%%%%%%%%%%%%%%%%%%%%%%%%%%%%%%%%%%%%%%%%%%%%%%%

Identities \eqref{fskbrz} immediately give rise to the identities
\begin{equation} \label{rksnf}
\frac{\sin\alpha}{ \sqrt{\gamma_a^2-1}}
= \frac{\sin\beta }{ \sqrt{\gamma_b^2-1}}
= \frac{\sin\gamma}{ \sqrt{\gamma_c^2-1}}
\end{equation}
that form the law of gyrosines.

Unlike Euclidean triangles, the side gyrolengths of a gyrotriangle are
uniquely determined by its gyroangles, as the following theorem demonstrates.
%%%%%%%%%%%%%%%%%%%%%%%%%%%%%%%%%%%%%%%%%%%%%%%%%%%%%%%%%%%%%%%%%%%%
% THEOREM NUMBER 12.3
% THEOREM NUMBER 4.11
\begin{ttheorem}\label{angtoside}
{\bf (The $AAA$ to $SSS$ Conversion Law).}
Let $ABC$ be a gyrotriangle
in an Einstein gyrovector space $(\VS,\op,\od)$.
Then, in the gyrotriangle notation in Fig.~\ref{fig171ec4m},
%%%%%%%%%%%%%%%%%%%%%%%%%%%%%%%%%%%%%%%%%%%%%%%%%%%%%%%%%%%%%%%%%%%%
\begin{equation} \label{jsf04}
\begin{split}
\gamma_a &= \frac{\cos\alpha+\cos\beta\cos\gamma}{\sin\beta\sin\gamma}\\[4pt]
\gamma_b &= \frac{\cos\beta+\cos\alpha\cos\gamma}{\sin\alpha\sin\gamma}\\[4pt]
\gamma_c &= \frac{\cos\gamma+\cos\alpha\cos\beta}{\sin\alpha\sin\beta}
\end{split}
\end{equation}
%%%%%%%%%%%%%%%%%%%%%%%%%%%%%%%%%%%%%%%%%%%%%%%%%%%%%%%%%%%%%%%%%%%%
\end{ttheorem}
\begin{proof}
%%%%%%%%%%%%%%%%%%%%%%%%%%%%%%%%%%%%%%%%%%%%%%%%%%%%%%%%%%%%%%%%%%%%
Let $ABC$ be a gyrotriangle in an Einstein gyrovector space
$(\VS,\od,\op)$ with its standard notation in Fig.~\ref{fig171ec4m}.
It follows straightforwardly from the
$SSS$ to $AAA$ conversion law \eqref{jsf02} that
\begin{equation} \label{jsf03}
\left(
\frac{\cos\alpha+\cos\beta\cos\gamma}{\sin\beta\sin\gamma}
\right)^2
=
\frac{(\cos\alpha+\cos\beta\cos\gamma)^2}{(1-\cos^2\beta)(1-\cos^2\gamma)}
= \gamma_a^2
\end{equation}
%MATHEMATICA stam028
implying the first identity in \eqref{jsf04}.
The remaining two identities in \eqref{jsf04} are obtained
from \eqref{jsf02} by permutation of vertices.
\end{proof} 
%%%%%%%%%%%%%%%%%%%%%%%%%%%%%%%%%%%%%%%%%%%%%%%%%%%%%%%%%%%%%%%%%%%%

The identities in \eqref{jsf04} form the $AAA$ to $SSS$ conversion law.
This law is useful for calculating the sides
(that is, the side-gyrolengths)
of a gyrotriangle in an Einstein gyrovector space when its gyroangles are known.
Thus, for instance, $\gamma_a$ is obtained from the first identity in \eqref{jsf04},
and $a$ is obtained from $\gamma_a$ by Identity \eqref{tksnw}.

%%%%%%%%%%%%%%%%%%%%%%%%%%%%%%%%%%%%%%%%%%%%%%%%%%%%%%%%%%%%%%%%%%%%
 % Fig. 3 of paper034_defect.tex
%% Fig 6.16 in Kluwer's book "beyond" %%%%%%%%%%%%%%%%%%%%%%%%%%%%%%%%%%%%%%
\begin{figure}[t]  % try to put this figure on the top of the page
 \centering         % center the figure
\psfrag{pa}{$\hspace{-0.1cm} A $}
\psfrag{pb}{$ B $}
\psfrag{pc}{$\hspace{-0.3cm} C $}
\psfrag{paa}{$\hspace{0.1cm}\ab$}
\psfrag{pbb}{$\bb$}
\psfrag{pcc}{$\cb$}
\psfrag{alfa}{$\alpha$}
\psfrag{beta}{$\beta$}
\psfrag{gama}{$\gamma=\pi/2$}
\psfrag{f1}{$\ab=\om B \op C  \hspace{0.6cm}a=\|\ab\|$}
\psfrag{f2}{$\bb=\om C \op A  \hspace{0.6cm}b=\|\bb\|$}
\psfrag{f3}{$\cb=\om A \op B  \hspace{0.64cm}c=\|\cb\|$}
%%%%%%%%%%%%%%%%%%%%%%%%%%%%%%%%%%%%%%%%%%%%%%%%%%%%%%%%%%%%
%\psfrag{f4}{$\sin\alpha=\displaystyle\frac{a_{_\gamma}}{c_{_\gamma}}
\psfrag{f4}{$\sin\alpha
=\displaystyle\frac{\gamma_a a}{\gamma_c c}\,,$}
%\psfrag{f5}{$\sin\beta =\displaystyle\frac{b_{_\gamma}}{c_{_\gamma}}
\psfrag{f5}{$\sin\beta 
=\displaystyle\frac{\gamma_b b}{\gamma_c c}\,,$}
 \psfrag{g6}{$\gamma_c=\gamma_a\gamma_b$}
 \psfrag{f6}{$\cos\alpha
  =\displaystyle\frac{b}{c}
                        =\displaystyle\frac{\bb\ccdot\cb}{bc}\,,$}
%\psfrag{f7}{$\cos\beta =\displaystyle\frac{a_{_\beta}}{c_{_\beta}}
 \psfrag{f7}{$\cos\beta
  =\displaystyle\frac{a}{c}
                        =\displaystyle\frac{\ab\ccdot\cb}{ac}\,,$}
 \psfrag{f8}{$\sin^2\alpha + \cos^2\alpha = 1$}
 \psfrag{f9}{$\sin^2\beta  + \cos^2\beta  = 1$}
%%%%%%%%%%%%%%%%%%%%%%%%%%%%%%%%%%%%%%%%%%%%%%%%%%%%%%%%%%%%
%\psfrag{ff}{$\hspace{-1.2cm}\delta=\frac{\pi}{2}-(\alpha+\beta),
%             \hspace{0.30cm}\tan\frac{\delta}{2}=(\half\od a)(\half\od b)/s^2$}
%\psfrag{gg}[]{$\boxed{\frac{
%(\half\od c)^2}{s}=\frac{(\half\od a)^2}{s}\op \frac{(\half\od b)^2}{s}}$}
%\includegraphics[width=9cm]{/home/ungar/dir_amy/dir_papers/dir_mybook01/dir_figs/fig149ein04d.eps}
 \includegraphics[width=9cm]{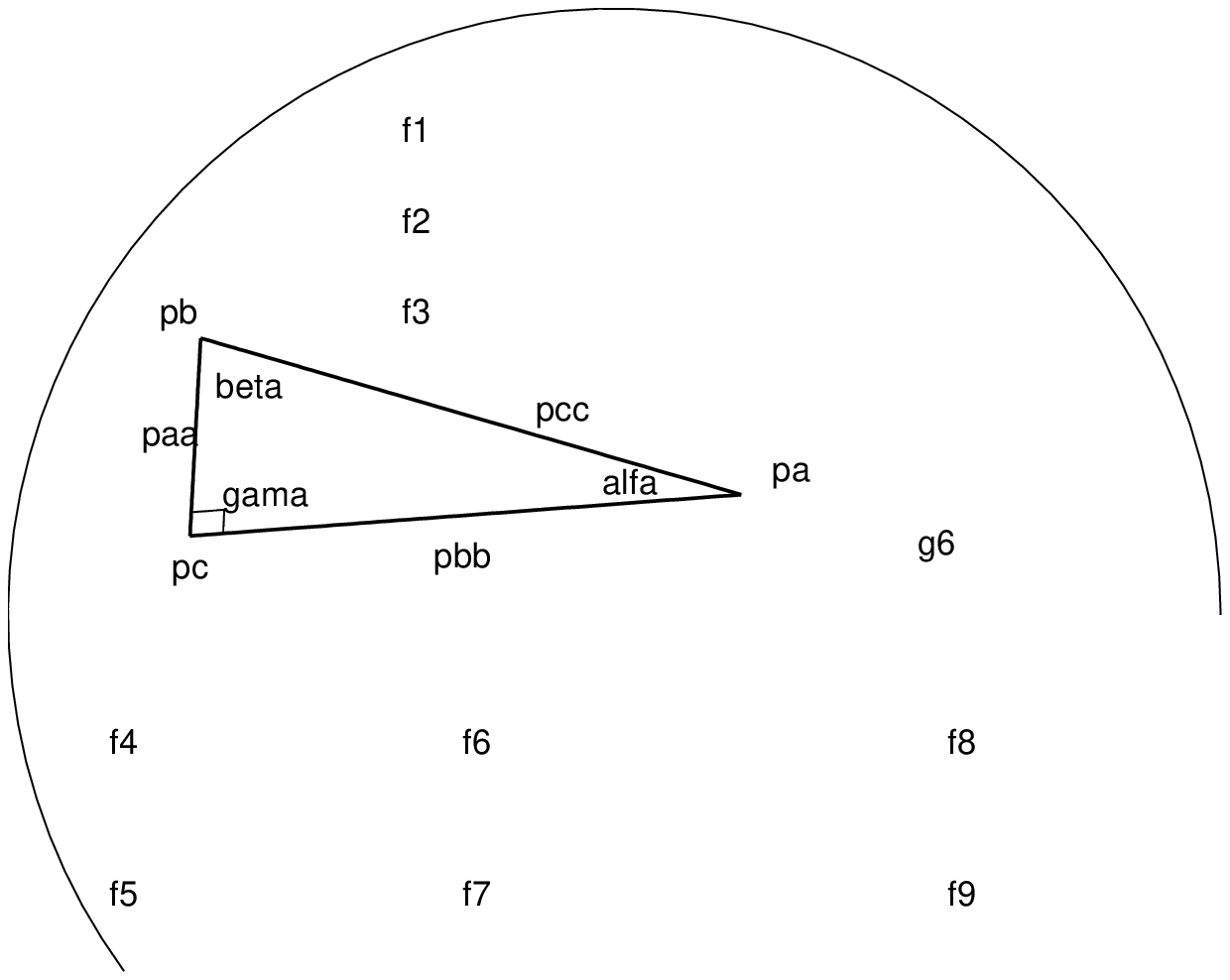}
\caption{
Gyrotrigonometry in an Einstein gyrovector plane $(\Rstwo,\op,\od)$.
\label{fig149ein04dm}}
\end{figure}
%%%%%%%%%%%%%%%%%%%%%%%%%%%%%%%%%%%%%%%%%%%%%%%%%%%%%%%%%%%%%%%%%%%%%%%%%%%%%
         %Figure - Einstein right gyrotriangle
                     % Fig. 7
                     %                         Fig.~\ref{fig149ein04dm}
%%%%%%%%%%%%%%%%%%%%%%%%%%%%%%%%%%%%%%%%%%%%%%%%%%%%%%%%%%%%%%%%%%%%

Let $ABC$ be a right gyrotriangle in an Einstein gyrovector space
$(\VS,\op,\od)$ with the right gyroangle $\gamma=\pi/2$,
as shown in Fig.~\ref{fig149ein04dm} for $\VS=\Rstwo$.
It follows from \eqref{jsf04} with $\gamma=\pi/2$ that the
sides $a,b$ and $c$ of gyrotriangle $ABC$ in Fig.~\ref{fig149ein04dm}
are related to the acute gyroangles
$\alpha$ and $\beta$ of the gyrotriangle by the equations
%%%%%%%%%%%%%%%%%%%%%%%%%%%%%%%%%%%%%%%%%%%%%%%%%%%%%%%%%%%%%%%%%%%%
\begin{equation} \label{jsf04a}
\begin{split}
\gamma_a^{\phantom{O}} &= \frac{\cos\alpha}{\sin\beta}\\[6pt]
\gamma_b^{\phantom{O}} &= \frac{\cos\beta}{\sin\alpha}\\[6pt]
\gamma_c^{\phantom{O}} &= \frac{\cos\alpha\cos\beta}{\sin\alpha\sin\beta}
\end{split}
\end{equation}
%%%%%%%%%%%%%%%%%%%%%%%%%%%%%%%%%%%%%%%%%%%%%%%%%%%%%%%%%%%%%%%%%%%%

The identities in \eqref{jsf04a} imply
the Einstein-Pythagoras Identity
\begin{equation} \label{jsf04b}
\gamma_a^{\phantom{O}} \gamma_b^{\phantom{O}} = \gamma_c^{\phantom{O}}
\end{equation}
for a right gyrotriangle $ABC$
with hypotenuse $c$ and legs $a$ and $b$
in an Einstein gyrovector space, Fig.~\ref{fig149ein04dm}.

Let $a,~b$ and $c$ be the respective gyrolengths of the
two legs $\ab,~\bb$ and the hypotenuse $\cb$ of a right gyrotriangle $ABC$
in an Einstein gyrovector space $(\VS,\op,\od)$, Fig.~\ref{fig149ein04dm}.
By \eqref{tksnw}, p.~\pageref{tksnw},
and \eqref{jsf04a} we have
%%%%%%%%%%%%%%%%%%%%%%%%%%%%%%%%%%%%%%%%%%%%%%%%%%%%%%%%%%%%%%%%%%%%
 \begin{equation} \label{ksnre}
 \begin{split}
 \left(\frac{a}{c}\right)^2 &= \frac{(\gamma_a^2-1)/\gamma_a^2}
                                    {(\gamma_c^2-1)/\gamma_c^2}
 =\cos^2\beta
 \\[8pt]
 \left(\frac{b}{c}\right)^2 &= \frac{(\gamma_b^2-1)/\gamma_b^2}
                                    {(\gamma_c^2-1)/\gamma_c^2}
 =\cos^2\alpha
 \end{split}
 \end{equation}
%%%%%%%%%%%%%%%%%%%%%%%%%%%%%%%%%%%%%%%%%%%%%%%%%%%%%%%%%%%%%%%%%%%%
where $\gamma_a^{\phantom{O}}$, $\gamma_b^{\phantom{O}}$ and $\gamma_c^{\phantom{O}}$
are related by \eqref{jsf04b}.

Similarly, by \eqref{tksnw},
and \eqref{jsf04a} we also have
%%%%%%%%%%%%%%%%%%%%%%%%%%%%%%%%%%%%%%%%%%%%%%%%%%%%%%%%%%%%%%%%%%%%
 \begin{equation} \label{ksnrd}
 \begin{split}
 \left(\frac{\gamma_a a}{\gamma_c c}\right)^2
 &= \frac{\gamma_a^2-1}{\gamma_c^2-1}
 =\sin^2\alpha
 \\[8pt]
 \left(\frac{\gamma_b b}{\gamma_c c}\right)^2
 &= \frac{\gamma_b^2-1}{\gamma_c^2-1}
 =\sin^2\beta
 \end{split}
 \end{equation}
%Calculated in MATH trigo015
%%%%%%%%%%%%%%%%%%%%%%%%%%%%%%%%%%%%%%%%%%%%%%%%%%%%%%%%%%%%%%%%%%%%

Identities \eqref{ksnre} and \eqref{ksnrd} imply
\begin{equation} \label{ksnrd01}
\begin{split}
 \left(\frac{a}{c}\right)^2 +
\left(\frac{\gamma_b b}{\gamma_c c}\right)^2
&=1 \\[8pt]
\left(\frac{\gamma_a a}{\gamma_c c}\right)^2 +
 \left(\frac{b}{c}\right)^2
&=1
\end{split}
 \end{equation}
and, as shown in Fig.~\ref{fig149ein04dm},
\begin{equation} \label{ksnrd02}
\begin{split}
\cos\alpha &= \frac{b}{c}
\\[8pt]
\cos\beta  &= \frac{a}{c}
\end{split}
 \end{equation}
and
\begin{equation} \label{ksnrd03}
\begin{split}
\sin\alpha  &= \frac{\gamma_a a}{\gamma_c c}
\\[8pt]
\sin\beta   &= \frac{\gamma_b b}{\gamma_c c}
\end{split}
 \end{equation}

Interestingly, we see from \eqref{ksnrd02}\,--\,\eqref{ksnrd03}
that the gyrocosine function of an acute gyroangle
of a right gyrotriangle in an Einstein gyrovector space
has the same form as its Euclidean counterpart, the cosine function.
In contrast, it is only modulo gamma factors that
the gyrosine function has the same form as its Euclidean counterpart,
the sine function.

Identities \eqref{ksnrd01}
give rise to the following two distinct
Einsteinian-Pythagorean identities,
\begin{equation} \label{ktnrd01}
\begin{split}
a^2 + \left(\frac{\gamma_b}{\gamma_c}\right)^2 b^2 &= c^2 \\[8pt]
\left(\frac{\gamma_a}{\gamma_c}\right)^2 a^2 + b^2 &= c^2
\end{split}
 \end{equation}
for a right gyrotriangle
with hypotenuse $c$ and legs $a$ and $b$ in an Einstein gyrovector space.
The two distinct Einsteinian-Pythagorean identities
in \eqref{ktnrd01} that each Einsteinian right gyrotriangle possesses
converge in the Newtonian-Euclidean limit of large $s$,
$s\rightarrow\infty$,
to the single Pythagorean identity
\begin{equation} \label{ktgrd01}
a^2 + b^2 = c^2
\end{equation}
that each Euclidean right-angled triangle possesses.

% SECTION NUMBER 7
\section{The Interpretation of\\
Classical Stellar Aberration by\\ Vectors and Trigonometry\\
$~~~~~~~~~~~~~~~~~~~~~~~~~~~~~~$
and\\
The Interpretation of\\
Relativistic Stellar Aberration by\\ Gyrovectors and Gyrotrigonometry}\label{secc7}

The universe is our laboratory, we are the experimenters asking nature whether,
in the limit of negligible force, relativistic velocities add
%%%%%%%%%%%%%%%%%%%%%%%%%%%%%%%%%%%%%%%%%%%%%%%%%%%%%%%%%%%%%%%%%%%%%%%%%%%%%
\begin{itemize}
\item[$\phantom{i}$]
\begin{itemize}
\item[$(1)$]
according to Einstein's 1905 velocity addition law \eqref{eq01}, or
\item[$(2)$]
according to the Einstein gyroparallelogram addition law \eqref{ktfr4}.
\end{itemize}
\end{itemize}
%%%%%%%%%%%%%%%%%%%%%%%%%%%%%%%%%%%%%%%%%%%%%%%%%%%%%%%%%%%%%%%%%%%%%%%%%%%%%

Fortunately, the cosmological phenomenon of {\it stellar aberration} comes to the rescue,
as we will see in this section.
We will find that owing to the
validity of well-known relativistic stellar aberration formulas,
%%%%%%%%%%%%%%%%%%%%%%%%%%%%%%%%%%%%%%%%%%%%%%%%%%%%%%%%%%%%%%%%%%%%%%%%%%%%%
\begin{itemize}
\item[$\phantom{i}$]
\begin{itemize}
\item[$(i)$]
Einsteinian, relativistic velocities are gyrovectors that add
according to the gyroparallelogram addition law \eqref{ktfr4}, Fig.~\ref{fig190k1m},
which is commutative, just as
\item[$(ii)$]
Newtonian, classical velocities are vectors that add
according to the common parallelogram addition law.
\end{itemize}
\end{itemize}
%%%%%%%%%%%%%%%%%%%%%%%%%%%%%%%%%%%%%%%%%%%%%%%%%%%%%%%%%%%%%%%%%%%%%%%%%%%%%

To set the stage for our study of the relativistic stellar aberration, we begin
with the study of the
{\it classical particle aberration}, Fig.~\ref{fig206sa4m}, which will be extended
to the study of the
{\it relativistic particle aberration}, Fig.~\ref{fig206sb4m},
by gyro-analogies.

%%%%%%%%%%%%%%%%%%%%%%%%%%%%%%%%%%%%%%%%%%%%%%%%%%%%%%%%%%%%%%%%%%%%
% This figure causes an error message! Why?
 
%%%%%%%%%%%%%%%%%%%%%%%%%%%%%%%%%%%%%%%%%%%%%%%%%%%%%%%%%%%%%%%%%%%%%%
%%%%  Figs 206sam and 206sbm %%%%%%%%%%%%%%%%%%%%%%%%%%%%%%%%%%%%%%%%%%%
%%%%%%%%%%%%%%%%%%%%%%%%%%%%%%%%%%%%%%%%%%%%%%%%%%%%%%%%%%%%%%%%%%%%%%
% The triangle velocity addition law: Classical (fig206sa4) and
%                                  Relativistic (fig206sb4).
%
\begin{figure}[t]  % try to put this figure on the top of the page
 \sidebyside {       % center two figures
%%%%%%%%%%%%%%%%%%%%%%%%%%%%%%%%%%%%%%%%%%%%%%%%%%%%%%%%%%%%%%%%%%%%
% The [] causes centering the LaTex box in the psfrag comand
% The [] can be deleted if not needed.
%%%%%%%%%%%%%%%%%%%%%%%%%%%%%%%%%%%%%%%%%%%%%%%%%%%%  Left Figure %%
%%%%%%%%%%%%%%%%%%%%%%%%%%%%%%%%%%%%%%%%%%%%%%%%%%%%%%%%%%%%%%%%%%%%
\psfrag{E}[]{E}
\psfrag{S}[]{S}
\psfrag{X}[]{Q}
\psfrag{P}[]{P}
\psfrag{O}[]{O}
\psfrag{pbs}[]{$\hspace{0.2cm}\pb_s$}
\psfrag{pbe}[]{$\hspace{0.32cm}\pb_e$}
\psfrag{vb}[]{$\leftarrow\vb$}
\psfrag{yb}[]{$\bb$}
\psfrag{xb}[]{$\ab$}
\psfrag{tetas}[]{$\hspace{-0.18cm}\theta_{\!s}$}
\psfrag{tetae}[]{$\hspace{-0.50cm}\theta_e$}
\psfrag{---A}{$\rightarrow$}
\psfrag{---B}{$\rightarrow$}
\psfrag{formula00}[]{$\vb= -  E  +  S \in \Rtwo$}
\psfrag{formula01}[]{\hspace{0.8cm} The velocity vector $\vb$,}
%%%%%%%%%%%%%%%%%%%%%%%%%%%%%%%%%%%%%%%%%%%%%%%%%%%%%%%%%%%%%%%%%%%%
 \includegraphics[width=0.5\textwidth]{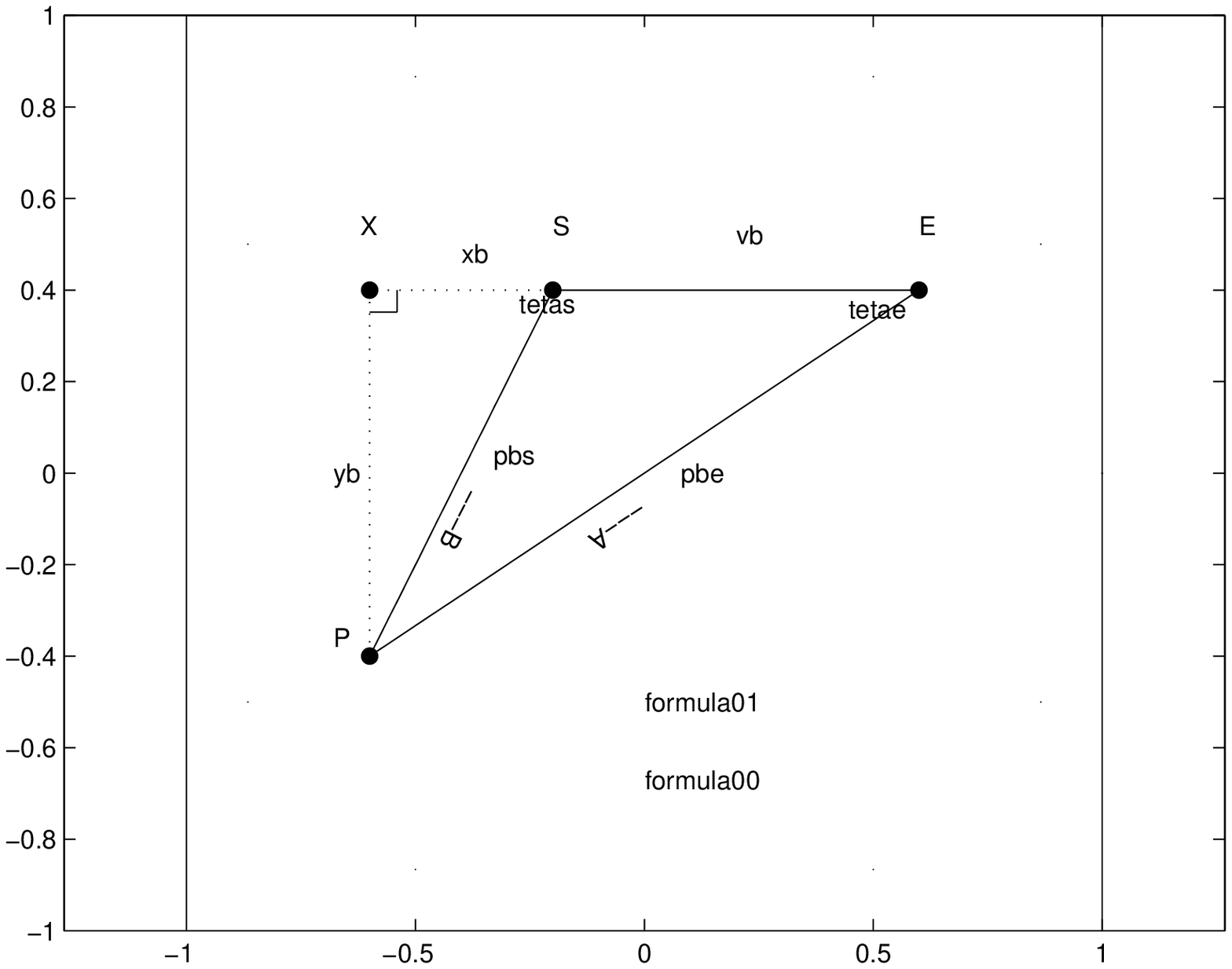}
%\includegraphics[width=0.4\textwidth]{/home/ungar/dir_amy/dir_papers/dir_mybook01/dir_figs/fig206sa4.eps}
% Left Figure
\caption[Particle Aberration: Classical]{
Classical Particle Aberration is commonly studied by velocity vectors and angles relative to
Cartesian coordinates $(x,y)$, $-\infty<x,y<\infty$,
in the Newtonian velocity space $\Rtwo$.
The Euclidity of $(\Rtwo,+)$ is determined by the
Euclidean metric in which the distance between two points $A,B$
is $\|-A+B\|$.
\label{fig206sa4m}}}
%%%%%%%%%%%%%%%%%%%%%%%%%%%%%%%%%%%%%%%%%%%%%%%%%%%%  Rite Figure %%
  {
%%%%%%%%%%%%%%%%%%%%%%%%%%%%%%%%%%%%%%%%%%%%%%%%%%%%%%%%%%%%%%%%%%%%%%
\psfrag{E}[]{E}
\psfrag{S}[]{S}
\psfrag{X}[]{Q}
\psfrag{P}[]{P}
\psfrag{O}[]{O}
\psfrag{pbs}[]{$\hspace{0.2cm}\pb_s$}
\psfrag{pbe}[]{$\hspace{0.32cm}\pb_e$}
\psfrag{vb}[]{$\leftarrow\vb$}
\psfrag{yb}[]{$\bb$}
\psfrag{xb}[]{$\ab$}
\psfrag{tetas}[]{$\hspace{-0.18cm}\theta_{\!s}$}
\psfrag{tetae}[]{$\hspace{-0.5cm}\theta_e$}
\psfrag{---A}{$\rightarrow$}
\psfrag{---B}{$\rightarrow$}
\psfrag{formula00}[]{$\hspace{0.14cm}\vb=\om E \op S \in \Rctwou$}
\psfrag{formula01}[]{\hspace{0.8cm} The velocity gyrovector $\vb$,}
%%%%%%%%%%%%%%%%%%%%%%%%%%%%%%%%%%%%%%%%%%%%%%%%%%%%%%%%%%%%%%%%%%%%%%
 \includegraphics[width=0.5\textwidth]{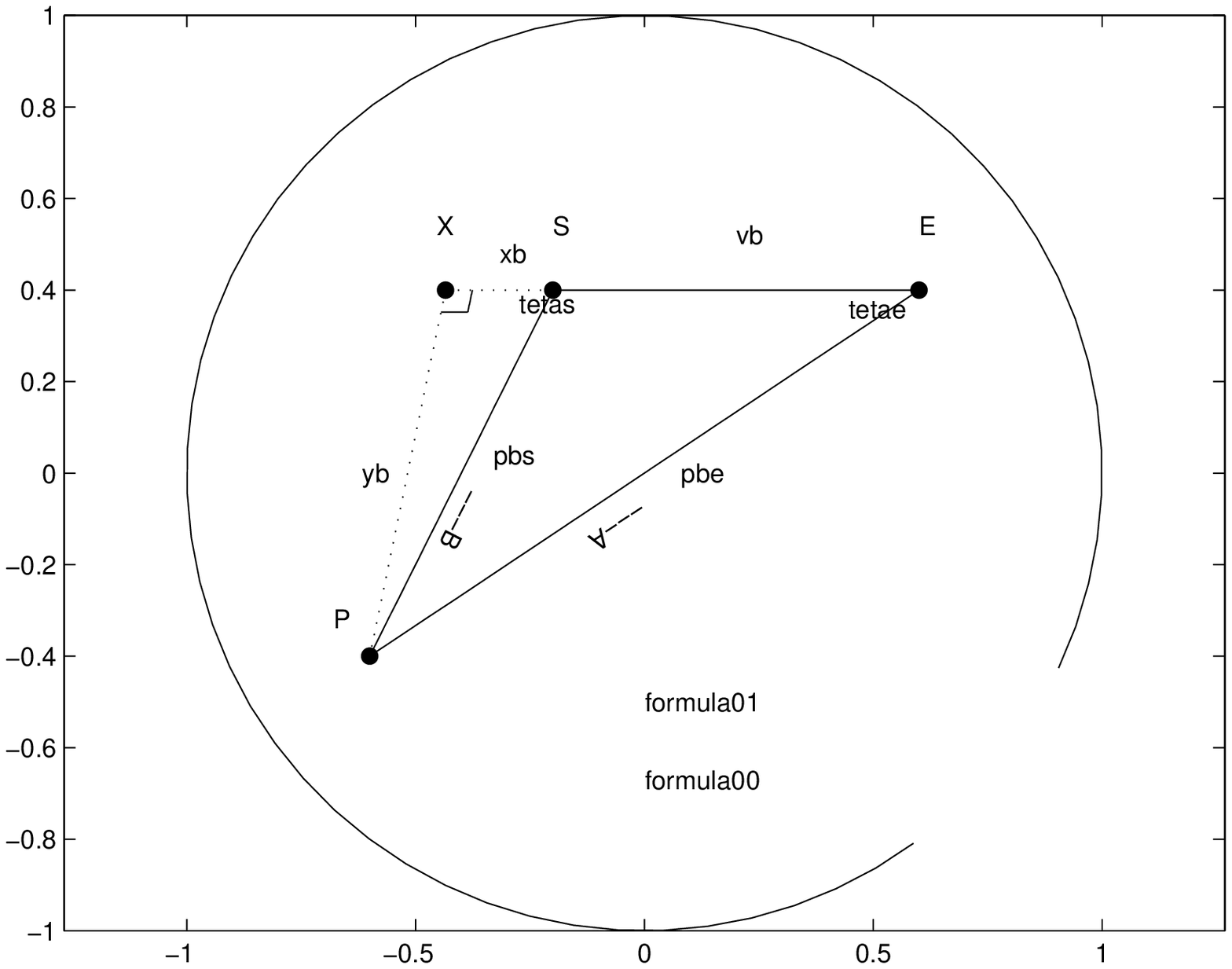}
%\includegraphics[width=0.4\textwidth]{/home/ungar/dir_amy/dir_papers/dir_mybook01/dir_figs/fig206sb4.eps}
% Right Figure
\caption[Particle Aberration: Relativistic]{
Relativistic Particle Aberration can be studied by velocity gyrovectors and gyroangles relative to
Cartesian coordinates $(x,y)$, $x^2+y^2<c^2=1$,
in the Einsteinian velocity space $\Rctwou$.
The hyperbolicity of $(\Rctwo,\op)$ is determined by the
hyperbolic gyrometric in which the distance between two points $A,B$
is $\|\om A\op B\|$.
\label{fig206sb4m}} }
\end{figure}
%%%%%%%%%%%%%%%%%%%%%%%%%%%%%%%%%%%%%%%%%%%%%%%%%%%%%%%%%%%%%%%%%%%%%%

          %Particle Aberration (1) Classical (2) Relativistic
                     % Figs. 8,9
                     %                         Fig.~\ref{fig206sa4m}
                     %                         Fig.~\ref{fig206sb4m}
%%%%%%%%%%%%%%%%%%%%%%%%%%%%%%%%%%%%%%%%%%%%%%%%%%%%%%%%%%%%%%%%%%%%

Figs.~\ref{fig206sa4m} and \ref{fig206sb4m}
present, respectively, the
Newtonian velocity space $\Rt$, and the
Einsteinian velocity space $\Rct$, along with several of their points, where only
two dimensions are shown for clarity.
The origin, $O$, of each of these two velocity spaces, not shown in
Figs.~\ref{fig206sa4m}\,--\,\ref{fig206sb4m},
is arbitrarily selected, representing an arbitrarily selected inertial rest frame
$\Sigma_\zerb$.

Points of a velocity space represent uniform velocities relative to the
rest frame $\Sigma_\zerb$.
In particular,
the points $E$, $S$ and $P$ in Figs.~\ref{fig206sa4m}\,--\,\ref{fig206sb4m}
represent, respectively, the velocity of the Earth, the Sun, and a Particle
(emitted, for instance, from a star) relative to the rest frame $\Sigma_\zerb$.

Accordingly, the Newtonian velocity vector $\vb$ of the Sun relative to the Earth is
%%%%%%%%%%%%%%%%%%%%%%%%%%%%%%%%%%%%%%%%%%%%%%%%%%%%%%%%%%%%%%%%%%%%%%%%%%%%%
% Produces equations in an array, numbered mylabela, mylabelb, mylabelc, etc.
\begin{subequations} \label{eq1}
\begin{eqnarray} \label{eq1a}
\vb &= -~E~+~S
\end{eqnarray}
and the Einsteinian velocity gyrovector $\vb$ of the Sun relative to the Earth is
\begin{eqnarray} \label{eq1b}
\vb &= \om~E~\op~S
\end{eqnarray}
\end{subequations}
%%%%%%%%%%%%%%%%%%%%%%%%%%%%%%%%%%%%%%%%%%%%%%%%%%%%%%%%%%%%%%%%%%%%%%%%%%%%%
as shown, respectively, in Figs.~\ref{fig206sa4m}\,--\,\ref{fig206sb4m}.

Similarly, the particle $P$ moves uniformly relative to the Earth and
relative to the Sun with respective Newtonian velocities, Fig.~\ref{fig206sa4m},
%%%%%%%%%%%%%%%%%%%%%%%%%%%%%%%%%%%%%%%%%%%%%%%%%%%%%%%%%%%%%%%%%%%%%%%%%%%%%
\begin{subequations} \label{eq2}
\begin{equation} \label{eq2a}
\begin{split}
\Pb_e &= -~E~+~P \\
\Pb_s &= -~S~+~P
\end{split}
\end{equation}
%%%%%%%%%%%%%%%%%%%%%%%%%%%%%%%%%%%%%%%%%%%%%%%%%%%%%%%%%%%%%%%%%%%%%%%%%%%%%

In full analogy with \eqref{eq2a},
the Einsteinian  velocities of the particle $P$ relative to the Earth and
relative to the Sun are, Fig.~\ref{fig206sb4m},
%%%%%%%%%%%%%%%%%%%%%%%%%%%%%%%%%%%%%%%%%%%%%%%%%%%%%%%%%%%%%%%%%%%%%%%%%%%%%
\begin{equation} \label{eq2b}
\begin{split}
\Pb_e &= \om~E~\op~P \\
\Pb_s &= \om~S~\op~P
\end{split}
\end{equation}
\end{subequations}
%%%%%%%%%%%%%%%%%%%%%%%%%%%%%%%%%%%%%%%%%%%%%%%%%%%%%%%%%%%%%%%%%%%%%%%%%%%%%

The Newtonian velocities \eqref{eq2a} of the particle
make angles $\theta_e$ and $\theta_s$,
respectively, with the Newtonian velocity $\vb$ in \eqref{eq1a}, as shown in
Fig.~\ref{fig206sa4m}.

Similarly,
the Einsteinian velocities \eqref{eq2b} of the particle
make gyroangles $\theta_e$ and $\theta_s$,
respectively, with the Einsteinian velocity $\vb$ in \eqref{eq1b}, as shown in
Fig.~\ref{fig206sb4m}.

Following Fig.~\ref{fig206sa4m},
{\it classical particle aberration} is the angular change $\theta_s - \theta_e$
in the apparent direction of a moving particle caused by the
motion with Newtonian relative velocity $\vb$, \eqref{eq1a}, between $E$ and $S$.
A relationship between the angles $\theta_s$ and $\theta_e$ is called a
{\it classical particle aberration formula}.

Similarly, following Fig.~\ref{fig206sb4m},
{\it relativistic particle aberration} is the gyroangular change $\theta_s - \theta_e$
in the apparent direction of a moving particle caused by the
motion with Einsteinian relative velocity $\vb$, \eqref{eq1b}, between $E$ and $S$.
A relationship between the gyroangles $\theta_s$ and $\theta_e$ is called a
{\it relativistic particle aberration formula}.

In order to uncover classical particle aberration formulas
we draw the altitude $PQ$ from vertex $P$ to side $ES$
(extended if necessary) obtaining the right-angled triangle
$EQP$ in Fig.~\ref{fig206sa4m}. The latter, in turn, enables
the parallelogram law and the triangle law of Newtonian velocity addition, and
the triangle equality and trigonometry, to be applied, obtaining the
following two, mutually equivalent, classical particle aberration formulas,
%%%%%%%%%%%%%%%%%%%%%%%%%%%%%%%%%%%%%%%%%%%%%%%%%%%%%%%%%%%%%%%%%%%%
\begin{subequations} \label{eqhgsd}
%%%%%%%%%%%%%%%%%%%%%%%%%%%%%%%%%%%%%%%%%%%%%%%%%%%%%%%%%%%%%%%%%%%%
\begin{equation} \label{jdkfb06}
\begin{split}
\cot\theta_e &= \cot\theta_s + \frac{v}{p_s \sin\theta_s} \\[12pt]
\cot\theta_s &= \cot\theta_e - \frac{v}{p_e \sin\theta_e}
\end{split}
\end{equation}
%%%%%%%%%%%%%%%%%%%%%%%%%%%%%%%%%%%%%%%%%%%%%%%%%%%%%%%%%%%%%%%%%%%%
The resulting classical particle aberration formulas \eqref{jdkfb06}
are in full agreement with formulas available in the literature;
see, for instance, \cite[Eq.~(134), p.~147]{synge65}.

The details of obtaining the classical particle aberration formulas \eqref{jdkfb06},
illustrated in  Fig.~\ref{fig206sa4m}, are presented in
\cite[Chap.~13]{mybook03} and, hence, will not be presented here.

In full analogy,
in order to uncover relativistic particle aberration formulas
we draw the altitude $PQ$ from vertex $P$ to side $ES$
(extended if necessary) obtaining the right gyrotriangle
$EQP$ in Fig.~\ref{fig206sb4m}. The latter, in turn, enables
the gyrotriangle law and the gyroparallelogram law of Einsteinian velocity addition,
and the gyrotriangle equality and gyrotrigonometry, to be applied, obtaining the
following two, mutually equivalent, relativistic particle aberration formulas,
\begin{equation} \label{rjskw1}
\begin{split}
\cot\theta_e &=
\gamma_v^{\phantom{1}} (\cot\theta_s + \frac{v}{p_s\sin\theta_s} )
\\
\cot\theta_s &=
\gamma_v^{\phantom{1}} (\cot\theta_e - \frac{v}{p_e\sin\theta_e} )
%\cot\theta_e &=
%\gamma_v^{\phantom{1}} \frac{\cos\theta_s + v/p_s}{\sin\theta_s}
%\\
%\cot\theta_s &=
%\gamma_v^{\phantom{1}} \frac{\cos\theta_e - v/p_e}{\sin\theta_e}
\end{split}
\end{equation}
%%%%%%%%%%%%%%%%%%%%%%%%%%%%%%%%%%%%%%%%%%%%%%%%%%%%%%%%%%%%%%%%%%%%
\end{subequations}
%%%%%%%%%%%%%%%%%%%%%%%%%%%%%%%%%%%%%%%%%%%%%%%%%%%%%%%%%%%%%%%%%%%%
The resulting relativistic particle aberration formulas \eqref{rjskw1}
are in full agreement with formulas available in the literature;
see, for instance,
\cite[p.~53]{rindler82},
\cite[p.~86]{rindler06}
and \cite[pp.~12--14]{landaulifshitz75}.

The details of obtaining the relativistic particle aberration formulas \eqref{rjskw1},
illustrated in  Fig.~\ref{fig206sb4m}, are presented in
\cite[Chap.~13]{mybook03} and, hence, will not be presented here.

In Euclidean geometry the triangle law and the parallelogram law of
vector addition are equivalent. In full analogy, their gyro-counterparts
are equivalent as well, as explained in \cite[Sec.~4.3]{mybook04} in detail.

%%%%%%%%%%%%%%%%%%%%%%%%%%%%%%%%%%%%%%%%%%%%%%%%%%%%%%%%%%%%%%%%%%%%
\begin{subequations}
%%%%%%%%%%%%%%%%%%%%%%%%%%%%%%%%%%%%%%%%%%%%%%%%%%%%%%%%%%%%%%%%%%%%
The equivalence of the two equations in \eqref{jdkfb06} implies
$p_s \sin\theta_s = p_e \sin\theta_e$, thus recovering
the law of sines\index{law of sines}
\begin{equation} \label{jdkge}
\frac{p_s}{\sin\theta_e} = \frac{p_e}{\sin\theta_s}
\end{equation}
for the Euclidean triangle $ESP$ in Fig.~\ref{fig206sa4m},
noting that $\sin\theta_s = \sin(\pi-\theta_s)$.

In full analogy, the equivalence of the two equations in \eqref{rjskw1} implies
the relativistic law of gyrosines \cite[Theorem 12.5]{mybook03},
\begin{equation} \label{kdnrw}
\frac{\gamma_{p_s}^{\phantom{1}} p_s}{\sin\theta_e}
=
\frac{\gamma_{p_e}^{\phantom{1}} p_e}{\sin\theta_s}
\end{equation}
% Confirmed numerically in MATLAB fig206sb.m "zerolaw=0".
%   for s=1; and in MATLAB test0385 for s>0.
for the gyrotriangle $ESP$ in Fig.~\ref{fig206sb4m}, noting that
$\sin\theta_s=\sin(\pi-\theta_s)$.
%%%%%%%%%%%%%%%%%%%%%%%%%%%%%%%%%%%%%%%%%%%%%%%%%%%%%%%%%%%%%%%%%%%%
\end{subequations}
%%%%%%%%%%%%%%%%%%%%%%%%%%%%%%%%%%%%%%%%%%%%%%%%%%%%%%%%%%%%%%%%%%%%

In this section we have described the way to recover the well-known
classical particle aberration formulas \eqref{jdkfb06}
by employing trigonometry, the triangle equality, the triangle addition law, and
the parallelogram addition law of Newtonian velocities.

In full analogy, we have described in this section the way to recover the well-known
relativistic particle aberration formulas \eqref{rjskw1}
by employing gyrotrigonometry, the gyrotriangle equality,
the gyrotriangle addition law, and
the gyroparallelogram addition law of Einsteinian velocities.

In contrast, the well-known
relativistic particle aberration formulas \eqref{rjskw1} are obtained in
the literature by employing the Lorentz transformation group of
special relativity.

What is remarkable here is that the
relativistic particle aberration formulas \eqref{rjskw1},
which are commonly obtained in the literature by Lorentz transformation considerations,
are recovered here by gyrotrigonometry and the
gyroparallelogram addition law of Einsteinian velocities,
in full analogy with the recovery of their classical counterparts.
This remarkable way of recovering the particle aberration formulas \eqref{eqhgsd}
demonstrates that
since special relativity is governed by the Lorentz transformation group,
Einsteinian velocities in special relativity
add according to the gyroparallelogram addition law,
just as Newtonian velocities add according to the parallelogram addition law.

Hence, any experiment that confirms the validity of the
relativistic particle aberration formulas \eqref{rjskw1}, amounts to an experiment
that confirms the validity of the gyroparallelogram addition law
of Einsteinian velocities.

In the special case when the particle $P$ in  Fig.~\ref{fig206sb4m}
is a photon emitted from a star, the Einsteinian speeds of the photon relative to both
$E$ and $S$ is $p_e = p_s = c$, and the
relativistic particle aberration formulas \eqref{rjskw1}
reduce to the corresponding {\it stellar aberration formulas},
%%%%%%%%%%%%%%%%%%%%%%%%%%%%%%%%%%%%%%%%%%%%%%%%%%%%%%%%%%%%%%%%%%%%
\begin{equation} \label{rjsfb1}
\begin{split}
\cot\theta_e &=
\gamma_v^{\phantom{1}} \frac{\cos\theta_s + v/c  }{\sin\theta_s}
\\
\cot\theta_s &=
\gamma_v^{\phantom{1}} \frac{\cos\theta_e - v/c  }{\sin\theta_e}
\end{split}
\end{equation}
%%%%%%%%%%%%%%%%%%%%%%%%%%%%%%%%%%%%%%%%%%%%%%%%%%%%%%%%%%%%%%%%%%%%
The discovery of stellar aberration, which results from the velocity of the Earth in
its annual orbit about the Sun, by the English astronomer James Bradley in the 1720s,
is described, for instance, in \cite{stewart64}.

A high precision test of the validity of the stellar aberration formulas \eqref{rjsfb1} in
special relativity has recently been obtained as a byproduct of the
``GP-B'' gyroscope experiment. Indeed,
the validity of the stellar aberration formulas \eqref{rjsfb1}
is central for the success of the
``GP-B'' gyroscope experiment developed by NASA and Stanford University
\cite{everitt69}
to test two unverified predictions of
Einstein's general theory of relativity \cite{everitt88,gp-b-URL1}.

The GP-B space gyroscopes encountered two kinds of stellar aberration.
Orbital aberration with 97.5-minute period of $\pm5.1856$ arc-seconds
that results from the motion of the gyroscopes
around the earth, and annual aberration with one year period
of about $\pm20.4958$ arc-seconds
that results from the motion
of the earth (and the gyroscopes) around the sun.
These aberrations, calculated by methods of special relativity,
were used to calibrate the gyroscopes and their
accompanying instruments.

If the ``GP-B'' gyroscope experiment proves successful, it could be considered
as an experimental evidence of the validity of the
stellar aberration formulas \eqref{rjsfb1}
and, hence, the validity of the
relativistic particle aberration formulas \eqref{rjskw1} as well.
The latter, in turn, could be considered
as an experimental evidence of the validity of the
gyroparallelogram addition law of Einsteinian velocities.
Indeed, the preliminary analysis of data has confirmed the theoretical
prediction of the ``GP-B'' gyroscope experiment \cite{gp-b-URL2},
so that the experiment seems to prove successful.

% SECTION NUMBER 8
\section{The Relativistic Mass and Dark Matter}\label{secc8}

Let
\begin{equation} \label{kyhd06}
S = S(m_k,\vb_k,\Sigma_\zerb,N)
\end{equation}
be an isolated system of
$N$ noninteracting material particles the $k$-th particle of which has
invariant mass $m_k >0$ and relativistically admissible velocity $\vb_k \inn\vc $
relative to an inertial frame $\Sigma_\zerb$,
$k=1,\dots,N$.

Assuming that the four-momentum is additive,
the sum of the four-momenta of the $N$ particles of the
system $S$ gives the four-momentum
$(m_0\gamma_{\vb_0}^{\phantom{O}}, m_0\gamma_{\vb_0}^{\phantom{O}} \vb_0)^t$
of $S$. Accordingly,
\begin{equation} \label{kyhd08}
\sum_{k=1}^{N} m_k
\begin{pmatrix}  \gamma_{\vb_k}^{\phantom{O}} \\[6pt]
\gamma_{\vb_k}^{\phantom{O}} \vb_k
\end{pmatrix}
=
m_0
\begin{pmatrix}  \gamma_{\vb_0}^{\phantom{O}} \\[6pt]
\gamma_{\vb_0}^{\phantom{O}} \vb_0
\end{pmatrix}
\end{equation}
where the invariant masses $m_k>0$ and the velocities $\vb_k$, $k=1,...,N$, relative
to $\Sigma_\zerb$ of the constituent particles of $S$ are given, while the unique
invariant mass $m_0$ of $S$ and the unique velocity $\vb_0$ of the CM frame of $S$
relative to $\Sigma_\zerb$ are to be determined.

In general, if $m_0$ in \eqref{kyhd08} exists, it satisfies the inequality
\begin{equation} \label{krhd07}
m_0 \ne \sum_{k=1}^{N} m_k
\end{equation}
leading to the conclusion that in special relativity
mass, $m_k$, is not additive \cite{marx91},
and relativistic mass, $m_k  \gamma_{\vb_k}^{\phantom{O}}$,
does not mesh up with the Minkowskian four-vector formalism of special relativity
\cite{okun89,brehme68,adler87}.

It follows immediately from \eqref{kyhd08} that if $m_0\ne0$ exists, then $\vb_0$
is given by the equation
\begin{equation} \label{hurfm}
\vb_0 = \frac{
\sum_{k=1}^{N} m_k \gamma_{\vb_k}^{\phantom{O}} \vb_k
}{
\sum_{k=1}^{N} m_k \gamma_{\vb_k}^{\phantom{O}}
}
\end{equation}

Employing the gyrocommutative gyrogroup structure of Einstein's velocity addition law,
it is found in \cite[Chap.~11]{mybook03}
that the unique solution of \eqref{kyhd08} for the unknown $m_0>0$
is given by the equation \cite{pioneer08,invariant09}
\begin{equation} \label{gksndh}
m_0 \phantom{i} = \phantom{i} \sqrt{
\left( \sum_{k=1}^{N} m_k \right)^2 +
2\sum_{\substack{j,k=1\\j<k}}^N m_j  m_k
(\gamma_{\om\vb_j\op\vb_k}^{\phantom{O}} -1)
}
\end{equation}

Hence, if the four-momentum is additive then, by \eqref{kyhd08}, the velocity $\vb_0$ of
the CM frame of $S$ relative to $\Sigma_\zerb$ is given by \eqref{hurfm},
and the invariant mass $m_0$ of $S$ is given by \eqref{gksndh}.

Furthermore, it follows from \eqref{kyhd08} that the relativistic mass
$m_0\gamma_{\vb_0}^{\phantom{O}}$ is additive, that is,
\begin{equation} \label{tkdnr}
m_0\gamma_{\vb_0}^{\phantom{O}} =
\sum_{k=1}^{N} m_k \gamma_{\vb_k}^{\phantom{O}}
\end{equation}

We thus see that owing to the introduction of the invariant mass $m_0$ of a
system of particles, given by \eqref{gksndh}, the relativistic mass is additive,
and it meshes extraordinarily well with the
Minkowskian four-vector formalism of special relativity.

Suggestively, we define the Newtonian mass, $m_{newton}$, of the system $S$
by the equation
\begin{equation} \label{fkns0i}
m_{newton} := \sum_{k=1}^{N} m_k
\end{equation}
and the dark mass, $m_{dark}$, of the system $S$
by the equation
\begin{equation} \label{fknsii}
m_{dark} := \sqrt{
2\sum_{\substack{j,k=1\\j<k}}^N m_j  m_k
(\gamma_{\om\vb_j\op\vb_k}^{\phantom{O}} -1)
}
\end{equation}
so that \eqref{gksndh} can be written as
\begin{equation} \label{tjansv}
m_0 = \sqrt{m_{newton}^2 + m_{dark}^2}
\end{equation}

The dark mass in \eqref{fknsii}
measures the extent to which the system $S$ deviates away from rigidity.
Gravitationally, dark mass behaves just like ordinary mass, as postulated
in cosmology \cite[p.~37]{conselice07}.
However, it is undetectable by all means other than gravity since it is
fictitious, or virtual,
in the sense that it is generated solely by relative motion
between constituent objects of the system.

Dark matter was introduced into cosmology as an {\it ad hoc} postulate,
hypothesized to provide observed missing gravitational force \cite{copeland06}.
In contrast, dark mass emerges here as a consequence of the
covariance of Einstein's special theory of relativity, and it stems from
relative motion between constituent objects of a system.
All relative velocities between the constituent particles of a
{\it rigid} system vanish, so that the dark mass of a rigid system vanishes
as well.

Under special circumstances dark matter may appear or disappear in galaxies,
a fact that may increase or decrease the total mass of galaxies which may,
in turn, decelerate or accelerate the expansion of the Universe.
These special circumstances are characterized by supernovae and star formation.

Each stellar explosion, a supernova, creates relative speeds between objects that
were at rest relative to each other prior to the explosion.
The resulting generated relative speeds increase the dark mass of
the region of the supernova.

Conversely, relative speeds
of objects that converge into a star vanish in the process of star formation,
resulting in the
decrease of the dark mass of a star formation region.

Dark mass is observed in particle physics as well.
Let us consider two particles with rest
masses $m_1$ and $m_2$, and velocities $\vb_1$ and $\vb_2$
relative to an inertial rest frame $\Sigma_\zerb$, respectively.
If these particles were to collide and stick, the rest mass $m_0$ and the
velocity $\vb_0$ relative to $\Sigma_\zerb$ of the resulting composite particle would
satisfy the four-momentum conservation law \eqref{kyhd08}, that is
%%%%%%%%%%%%%%%%%%%%%%%%%%%%%%%%%%%%%%%%%%%%%%%%%%%%%%%%%%%%%%%%%%%%%%%%%%%%%
\begin{equation} \label{fugnf1}
m_0
\begin{pmatrix}  \gamma_{\vb_0}^{\phantom{O}} \\[6pt]
\gamma_{\vb_0}^{\phantom{O}} \vb_0 \end{pmatrix}
=
m_1
\begin{pmatrix}  \gamma_{\vb_1}^{\phantom{O}} \\[6pt]
\gamma_{\vb_1}^{\phantom{O}} \vb_1 \end{pmatrix}
+
m_2
\begin{pmatrix}  \gamma_{\vb_2}^{\phantom{O}} \\[6pt]
\gamma_{\vb_2}^{\phantom{O}} \vb_2 \end{pmatrix}
\end{equation}
%%%%%%%%%%%%%%%%%%%%%%%%%%%%%%%%%%%%%%%%%%%%%%%%%%%%%%%%%%%%%%%%%%%%%%%%%%%%%

Hence, by \eqref{gksndh} and \eqref{tjansv},
\begin{equation} \label{fugnf2}
\begin{split}
m_0 &= \sqrt{
(m_1+m_2)^2 + 2m_1 m_2 (\gamma_{\om\vb_1\op\vb_2}^{\phantom{O}} - 1)
}
\\[6pt] &=  \sqrt{
m_{newton}^2 + m_{dark}^2
}
\end{split}
\end{equation}
where
\begin{equation} \label{fugnf3}
\begin{split}
m_{newton} &= m_1 + m_2 \\
m_{dark} &= 2m_1 m_2 (\gamma_{\om\vb_1\op\vb_2}^{\phantom{O}} - 1) > 0
\end{split}
\end{equation}
and, by \eqref{hurfm},
\begin{equation} \label{fugnf4}
\vb_0 = \frac{
m_1 \gamma_{\vb_1}^{\phantom{O}} \vb_1 +
m_2 \gamma_{\vb_2}^{\phantom{O}} \vb_2
}{
m_1 \gamma_{\vb_1}^{\phantom{O}} +
m_2 \gamma_{\vb_2}^{\phantom{O}}
}
\end{equation}
Hence, the relativistic mass of the composite particle is
$m_0\gamma_{\vb_0}^{\phantom{O}}$, where $m_0$ is given by \eqref{fugnf2},
and $\vb_0$ is given by \eqref{fugnf4}, satisfying by \eqref{tkdnr},
\begin{equation} \label{jughd}
m_0\gamma_{\vb_0}^{\phantom{O}} =
m_1\gamma_{\vb_1}^{\phantom{O}} +
m_2\gamma_{\vb_2}^{\phantom{O}}
\end{equation}

It is clear from \eqref{fugnf2}\,--\eqref{fugnf3}
that the Newtonian mass, $m_{newton}$, is conserved during the collision.
It is only the total invariant mass, $m_0$, which is increased following the collision owing to
the emergence of the dark mass $m_{dark}$.

Examples of particles that collide and stick, as described in
\eqref{fugnf1}\,--\,\eqref{jughd}, are observed in experimental searches for
new particles in high-energy particle colliders.
We thus see that our study of the dark mass may be applied on the subatomic scale
as well as on the scale of the cosmos.

%%%%%%%%  End of Paper069  %%%%%%%%%%%%%%%%%%%%%%%%%%%%%%%%%%%%%%%%%%%%%%%%%%%

\begin{thebibliography}{10}

\bibitem{adler87}
Carl~G. Adler.
\newblock Does mass really depend on velocity, dad?
\newblock {\em Amer. J. Phys.}, 55(8):739--743, 1987.

\bibitem{barrett01}
J.~F. Barrett.
\newblock On {C}arath\'eodory's approach to relativity and its relation to
  hyperbolic geometry.
\newblock In {\em Constantin Carath\'eodory in his $\ldots$ origins
  (Vissa-Orestiada, 2000)}, pages 81--90. Hadronic Press, Palm Harbor, FL,
  2001.

\bibitem{barrett00}
J.F. Barrett.
\newblock hyperbolic geometry in special relativity.
\newblock In {\em Recent Advances in Relativity Theory. Proceedings,
  {M.C.~D}uffy and {M.T. W}egener, eds.}, pages 27--34. Hadronic Press, Palm
  Harbor, FL, US, 2000.

\bibitem{barrett07}
J.F. Barrett.
\newblock Using relative velocities and hyperbolic geometry in special
  relativity.
\newblock In {\em Conf. {\it {P}hysical {I}nterpretations of {R}elativity
  {T}heory} (PIRT), Budapest}. 2007.

\bibitem{brehme68}
Robert~W. Brehme.
\newblock The advantage of teaching relativity with four-vectors.
\newblock {\em Amer. J. Phys.}, 36(10):896--901, 1968.

\bibitem{conselice07}
Christopher~J. Conselice.
\newblock The universe's invisible hand.
\newblock {\em Sci.~Amer.}, February:35--41, 2007.

\bibitem{copeland06}
Edmund~J. Copeland, M.~Sami, and Shinji Tsujikawa.
\newblock Dynamics of dark energy.
\newblock {\em Internat. J. Modern Phys. D}, 15(11):1753--1935, 2006.

\bibitem{einstein05}
Albert Einstein.
\newblock Zur {E}lektrodynamik {B}ewegter {K}\"orper [on the electrodynamics of
  moving bodies] ({W}e use the {E}nglish translation in \cite{einsteinfive} or
  in \cite{lorentz52}, or in
  http://www.fourmilab.ch/etexts/einstein/specrel/www/).
\newblock {\em Ann. Physik (Leipzig)}, 17:891--921, 1905.

\bibitem{einsteinfive}
Albert Einstein.
\newblock {\em Einstein's Miraculous Years: Five Papers that Changed the Face
  of Physics}.
\newblock Princeton, Princeton, NJ, 1998.
\newblock Edited and introduced by John Stachel. Includes bibliographical
  references. Einstein's dissertation on the determination of molecular
  dimensions -- Einstein on Brownian motion -- Einstein on the theory of
  relativity -- Einstein's early work on the quantum hypothesis. A new English
  translation of Einstein's 1905 paper on pp. 123--160.

\bibitem{everitt88}
C.W.F. Everitt.
\newblock {\em New frontiers in physics}.
\newblock Freeman, San Francisco, 1988.
\newblock in J.D.~Fairbank, B.S.~Deaver, Jr., C.W.F.~Everitt (eds.).

\bibitem{everitt69}
Francis C.~W. Everitt, William~M. Fairbank, and L.~I. Schiff.
\newblock {\em Theoretical background and present status of the Stanford
  relativity--gyroscope experiment}.
\newblock Switzerland, 1969.
\newblock in {\it The Significance of Space Research for Fundamental Physics},
  Proceedings of the Colloquium of the European Space Research Organization at
  Interlaken, Sept. 4, 1969. For current status see the project website at
  {URL}: http://einstein.stanford.edu.

\bibitem{feder03}
Tom{\'a}s Feder.
\newblock Strong near subgroups and left gyrogroups.
\newblock {\em J. Algebra}, 259(1):177--190, 2003.

\bibitem{fock}
V.~Fock.
\newblock {\em The theory of space, time and gravitation}.
\newblock The Macmillan Co., New York, 1964.
\newblock Second revised edition. Translated from the Russian by N. Kemmer. A
  Pergamon Press Book.

\bibitem{tuvalungar01}
Tuval Foguel and Abraham~A. Ungar.
\newblock Involutory decomposition of groups into twisted subgroups and
  subgroups.
\newblock {\em J. Group Theory}, 3(1):27--46, 2000.

\bibitem{tuvalungar02}
Tuval Foguel and Abraham~A. Ungar.
\newblock Gyrogroups and the decomposition of groups into twisted subgroups and
  subgroups.
\newblock {\em Pac. J. Math}, 197(1):1--11, 2001.

\bibitem{gp-b-URL1}
{GP-B}.
\newblock Gravity {P}robe {B}.
\newblock {\em \\ {\tt http://www.nasa.gov/mission\_pages/gpb/index.html}}.

\bibitem{gp-b-URL2}
{GP-B}.
\newblock Gravity {P}robe {B}.
\newblock {\em \\ {\tt http://einstein.stanford.edu}}.

\bibitem{issa99}
A.~Nourou Issa.
\newblock Gyrogroups and homogeneous loops.
\newblock {\em Rep. Math. Phys.}, 44(3):345--358, 1999.

\bibitem{issa2001}
A.~Nourou Issa.
\newblock Left distributive quasigroups and gyrogroups.
\newblock {\em J. Math. Sci. Univ. Tokyo}, 8(1):1--16, 2001.

\bibitem{kasparian04}
Azniv~K. Kasparian and Abraham~A. Ungar.
\newblock Lie gyrovector spaces.
\newblock {\em J. Geom. Symm. Phys.}, 1(1):3--53, 2004.

\bibitem{kikkawa75}
Michihiko Kikkawa.
\newblock Geometry of homogeneous {L}ie loops.
\newblock {\em Hiroshima Math. J.}, 5(2):141--179, 1975.

\bibitem{kikkawa99}
Michihiko Kikkawa.
\newblock Geometry of homogeneous left {L}ie loops and tangent {L}ie triple
  algebras.
\newblock {\em Mem. Fac. Sci. Eng. Shimane Univ. Ser. B Math. Sci.}, 32:57--68,
  1999.

\bibitem{kuznetsov03}
Eugene Kuznetsov.
\newblock Gyrogroups and left gyrogroups as transversals of a special kind.
\newblock {\em Algebra Discrete Math.}, (3):54--81, 2003.

\bibitem{landaulifshitz75}
L.~D. Landau and E.~M. Lifshitz.
\newblock {\em The classical theory of fields}.
\newblock Fourth revised English edition. Course of Theoretical Physics, Vol.
  2. Translated from the Russian by Morton Hamermesh. Pergamon Press, Oxford,
  New York, 1975.

\bibitem{lorentz52}
H.~A. Lorentz, A.~Einstein, H.~Minkowski, and H.~Weyl.
\newblock {\em The principle of relativity}.
\newblock Dover Publications Inc., New York, N. Y., undated.
\newblock With notes by A. Sommerfeld, Translated by W. Perrett and G. B.
  Jeffery, A collection of original memoirs on the special and general theory
  of relativity.

\bibitem{marx91}
George Marx.
\newblock Is the amount of matter additiv?
\newblock {\em European J. Phys.}, 12:271--274, 1991.

\bibitem{mccleary94}
John McCleary.
\newblock {\em Geometry from a differentiable viewpoint}.
\newblock Cambridge University Press, Cambridge, 1994.

\bibitem{miller81}
Arthur~I. Miller.
\newblock {\em Albert {E}instein's special theory of relativity}.
\newblock Springer-Verlag, New York, 1998.
\newblock Emergence (1905) and early interpretation (1905--11), Includes a
  translation by the author of Einstein's ``On the electrodynamics of moving
  bodies'', Reprint of the 1981 edition.

\bibitem{moller52}
C.~M{\o}ller.
\newblock {\em The theory of relativity}.
\newblock Oxford, at the Clarendon Press, 1952.

\bibitem{okun89}
Lev~B. Okun.
\newblock The concept of mass.
\newblock {\em Phys. Today}, 42(6):31--36, 1989.

\bibitem{pauli}
W.~Pauli.
\newblock {\em Theory of relativity}.
\newblock Pergamon Press, New York, 1958.
\newblock Translated from the German by G. Field, with supplementary notes by
  the author.

\bibitem{rassiasrev2008}
Th.~M. Rassias.
\newblock Book {R}eview: {\it {A}nalytic {H}yperbolic {G}eometry and {A}lbert
  {E}instein's {S}pecial {T}heory of {R}elativity}, by {A}braham {A}.~{U}ngar.
\newblock {\em Nonlinear Funct. Anal. Appl.}, 13(1):167--177, 2008.

\bibitem{rindler82}
Wolfgang Rindler.
\newblock {\em Introduction to special relativity}.
\newblock The Clarendon Press Oxford University Press, New York, 1982.

\bibitem{rindler06}
Wolfgang Rindler.
\newblock {\em Relativity: special, general, and cosmological}.
\newblock Oxford University Press, New York, second edition, 2006.

\bibitem{urbantkebookeng}
Roman~U. Sexl and Helmuth~K. Urbantke.
\newblock {\em Relativity, groups, particles}.
\newblock Springer Physics. Springer-Verlag, Vienna, 2001.
\newblock Special relativity and relativistic symmetry in field and particle
  physics, Revised and translated from the third German (1992) edition by
  Urbantke.

\bibitem{silberstein14}
L.~Silberstein.
\newblock {\em The Theory of Relativity}.
\newblock MacMillan, London, 1914.

\bibitem{stewart64}
Albert~B. Stewart.
\newblock The discovery of stellar aberration.
\newblock {\em Sci.~Amer.}, March:100--108, 1964.

\bibitem{stillwell96}
John Stillwell.
\newblock {\em Sources of hyperbolic geometry}.
\newblock American Mathematical Society, Providence, RI, 1996.

\bibitem{synge65}
J.~L. Synge.
\newblock {\em Relativity: the special theory}.
\newblock North-Holland Publishing Co., Amsterdam, second edition, 1965.

\bibitem{thomas26}
Llewellyn~H. Thomas.
\newblock The motion of the spinning electron.
\newblock {\em Nature}, 117:514, 1926.

\bibitem{parametrization}
Abraham~A. Ungar.
\newblock Thomas rotation and the parametrization of the {L}orentz
  transformation group.
\newblock {\em Found. Phys. Lett.}, 1(1):57--89, 1988.

\bibitem{grouplike}
Abraham~A. Ungar.
\newblock Thomas precession and its associated grouplike structure.
\newblock {\em Amer. J. Phys.}, 59(9):824--834, 1991.

\bibitem{gaxioms}
Abraham~A. Ungar.
\newblock Thomas precession: its underlying gyrogroup axioms and their use in
  hyperbolic geometry and relativistic physics.
\newblock {\em Found. Phys.}, 27(6):881--951, 1997.

\bibitem{mybook01}
Abraham~A. Ungar.
\newblock {\em Beyond the {E}instein addition law and its gyroscopic {T}homas
  precession: The theory of gyrogroups and gyrovector spaces}, volume 117 of
  {\em Fundamental Theories of Physics}.
\newblock Kluwer Academic Publishers Group, Dordrecht, 2001.

\bibitem{mybook02}
Abraham~A. Ungar.
\newblock {\em Analytic hyperbolic geometry: Mathematical foundations and
  applications}.
\newblock World Scientific Publishing Co. Pte. Ltd., Hackensack, NJ, 2005.

\bibitem{ungarthomas06}
Abraham~A. Ungar.
\newblock Thomas precession: a kinematic effect of the algebra of {E}instein's
  velocity addition law. {C}omments on: ``{D}eriving relativistic momentum and
  energy. {II}. {T}hree-dimensional case'' [{E}uropean {J}. {P}hys. {\bf 26}
  (2005), no. 5, 851--856; mr2227176] by {S}. {S}onego and {M}. {P}in.
\newblock {\em European J. Phys.}, 27(3):L17--L20, 2006.

\bibitem{mybook03}
Abraham~A. Ungar.
\newblock {\em Analytic hyperbolic geometry and {A}lbert {E}instein's special
  theory of relativity}.
\newblock World Scientific Publishing Co. Pte. Ltd., Hackensack, NJ, 2008.

\bibitem{mbtogyp08}
Abraham~A. Ungar.
\newblock From {M}\"obius to gyrogroups.
\newblock {\em Amer. Math. Monthly}, 115(2):138--144, 2008.

\bibitem{pioneer08}
Abraham~A. Ungar.
\newblock On the origin of the dark matter/energy in the universe and the
  {P}ioneer anomaly.
\newblock {\em Prog. Phys.}, 3:24--29, 2008.

\bibitem{mybook04}
Abraham~A. Ungar.
\newblock {\em A gyrovector space approach to hyperbolic geometry}.
\newblock Morgan \& Claypool Pub., San Rafael, California, 2009.

\bibitem{varicak10}
Vladimir Vari{\v c}ak.
\newblock Anwendung der {L}obatschefskjschen {G}eometrie in der
  {R}elativtheorie.
\newblock {\em Physikalische Zeitschrift}, 11:93--96, 1910.

\bibitem{varicak24}
Vladimir Vari{\v c}ak.
\newblock {\em Darstellung der {R}elativit\"atstheorie im dreidimensionalen
  {L}obatchefskijschen {R}aume [Presentation of the theory of relativity in the
  three-dimensional {L}obachevskian space]}.
\newblock Zaklada, Zagreb, 1924.

\bibitem{vermeer05}
J.~Vermeer.
\newblock A geometric interpretation of {U}ngar's addition and of gyration in
  the hyperbolic plane.
\newblock {\em Topology Appl.}, 152(3):226--242, 2005.

\bibitem{walter99b}
Scott Walter.
\newblock The non-{E}uclidean style of {M}inkowskian relativity.
\newblock In J.~J. Gray, editor, {\em The symbolic universe: Geometry and
  physics 1890--1930}, pages 91--127. Oxford Univ. Press, New York, 1999.

\bibitem{walterrev2002}
Scott Walter.
\newblock Book {R}eview: {\it {B}eyond the {E}instein {A}ddition {L}aw and its
  {G}yroscopic {T}homas {P}recession: {T}he {T}heory of {G}yrogroups and
  {G}yrovector {S}paces}, by {A}braham {A}. {U}ngar.
\newblock {\em Found. Phys.}, 32(2):327--330, 2002.

\bibitem{whittaker49}
Edmund~Taylor Whittaker.
\newblock {\em From {E}uclid to {E}ddington. {A} {S}tudy of {C}onceptions of
  the {E}xternal {W}orld}.
\newblock Cambridge University Press, 1949.

\end{thebibliography}
\end{document}